\documentclass[10pt]{article}

\title{Adaptive Channel Reshaping for Improved Entanglement Distillation}
\author{dina.abdelhadi}
\usepackage{caption}
\usepackage{pifont}
\usepackage{subcaption}
\usepackage[utf8]{inputenc}

\usepackage{amssymb}
\usepackage{amsfonts}
\usepackage{amsmath}
\usepackage{amsthm}
\usepackage{thmtools,thm-restate}
\usepackage{slashbox}
\usepackage{mathtools}
\usepackage{bbm,xfrac}
\usepackage{graphicx,float}
\usepackage{indentfirst}
\usepackage{enumerate}
\usepackage{url}

\usepackage{breakurl}
\usepackage{stmaryrd}
\usepackage{lmodern}
\usepackage{newpxtext}
\usepackage{algpseudocode}
\usepackage[symbol]{footmisc}

\usepackage[normalem]{ulem}

\usepackage{tikz}
\usetikzlibrary{quantikz2}
\definecolor{darkblue}{rgb}{0.15,0.35,0.55}
\definecolor{reddish}{rgb}{.8, 0.2, 0.2}

\usepackage{pgfplots}

\usetikzlibrary{plotmarks}% plot colors
\definecolor{plotblue}{RGB}{0,120,200}
\definecolor{plotgreen}{RGB}{0,155,130}
\definecolor{plotorange}{RGB}{240,120,50}
\definecolor{plotmagenta}{RGB}{240,50,120}
\definecolor{plotgray}{RGB}{128,128,128}
\definecolor{plotcyan}{RGB}{50,190,240}
\definecolor{plotred}{RGB}{205,50,15}
\usepgfplotslibrary{fillbetween}

\usepackage[backend=bibtex, style = alphabetic, doi = false, url = true, isbn =
false, maxbibnames = 6]{biblatex}
\bibliography{refs}
\renewbibmacro{in:}{}
\newbibmacro{string+doi}[1]{\iffieldundef{doi}{#1}{\href{https://dx.doi.org/\thefield{doi}}{#1}}}
\DeclareFieldFormat{title}{\usebibmacro{string+doi}{\mkbibemph{#1}}}
\DeclareFieldFormat[article]{title}{\usebibmacro{string+doi}{\mkbibquote{#1}}}
\DeclareFieldFormat[incollection]{title}{\usebibmacro{string+doi}{\mkbibquote{#1}}}
\DeclareFieldFormat[inproceedings]{title}{\usebibmacro{string+doi}{\mkbibquote{#1}}}

\usepackage[pdfpagelabels, linktocpage=true, breaklinks]{hyperref}
\hypersetup{
colorlinks=true,
citecolor=darkblue,
linkcolor=reddish,
urlcolor=darkblue,
pdfauthor={},
pdftitle={},
pdfsubject={}
}
\newcommand{\Tr}{{\rm Tr}}

\newcommand{\ketbra}[1]{\ket{#1}\bra{#1}}

\theoremstyle{remark}
\newtheorem{theorem}{Theorem}
\newtheorem{lemma}{Lemma}

\usepackage[linesnumbered]{algorithm2e}

\usetikzlibrary{arrows}
\usetikzlibrary{decorations.pathreplacing,calligraphy}
\usepackage{verbatim}

\usepackage{bbm}
 
\usepackage[affil-it]{authblk}

\author[1]{Dina Abdelhadi}
\author[2,3]{Tomas Jochym-O’Connor}
\author[2,4]{Vikesh Siddhu}
\author[2]{John Smolin}
\affil[1]{School of Computer and Communication Sciences, EPFL, Switzerland}
\affil[2]{IBM Quantum, IBM T.J. Watson Research Center, Yorktown Heights, NY, USA}
\affil[3]{IBM Quantum, Almaden Research Center, San Jose, CA, USA}
\affil[4]{IBM Research, IBM Research India, India}

\begin{document}

\maketitle
\begin{abstract}
Quantum communication and computation heavily rely on entanglement distillation protocols. There is a plethora of distillation protocols for Pauli channels and also for some non-Pauli channels. However, an effort to relate the effectiveness of these protocols has been missing. For most quantum channels, the gap between the existing lower and upper bounds on distillation rates is substantial, and improvements of achievable rates have been stagnant for decades. In this work, we improve the best known distillation lower bounds, for both the amplitude damping and depolarizing channels. We build on a key observation that distillation protocols reshape several uses of a very noisy channel into a better effective channel. We apply this channel processing in an adaptive and recurrent manner. For the amplitude damping channel, our suggested protocol reshapes the channel into an erasure channel, achieving rates exceeding the best known lower bound given by the channel's reverse coherent information. For the depolarizing channel, we introduce the Greedy recurrence protocol with proven performance guarantees and construct a combined protocol improving upon previously known distillation rates. Improved bounds on attainable distillation rates give insights for both practical implementations and theoretical understanding of quantum information processing.
\end{abstract}

\section{Introduction}
Useful quantum information processing systems need to share and distribute high
fidelity entangled states across noisy communication links. This link can be a
physical interconnect that joins computing or communication modules or a
sequence of noisy gates that entangle distant qubits on a single module. A noisy
point-to-point link is mathematically modeled by a completely positive trace
preserving map~(also called a quantum channel). Entanglement shared across a
quantum channel becomes noisy and entanglement distillation is the key protocol to
remove this noise~\cite{BennettBrassardEA96}.  

The ultimate benchmark for carrying out entanglement distillation across a
quantum channel is its quantum capacity: the number of Bell pairs per channel
use that can be shared between channel sender and receiver with asymptotically
vanishing error~\cite{BennettBrassardEA96, BennettShor98, BarnumKnillEA00, KretschmannWerner04}.
It is customary to assume availability of free noiseless classical messaging
for entanglement distillation since the cost of using classical channels is
much less than that of a quantum channel.

A channel's quantum capacity depends on whether one can send classical messages
one-way, from channel sender to receiver, or two-way, to and from sender and
receiver. In the former case we obtain the one-way quantum capacity $Q$ and in the
latter we obtain the two-way quantum capacity $Q_2$.
These capacities can be very different from one another~\cite{BennettBrassardEA96}, for instance a qubit
erasure channel with erasure probability $\epsilon$ has
$Q_2 - Q = \min(\epsilon, 1-\epsilon)$~\cite{Bennett_1997}.

A few main goals of quantum Shannon theory are to gain insight into channel
capacities, find mathematical expressions for computing them, and protocols to
achieve them. 
One such insight is that the one-way quantum capacity benchmarks quantum error
correction while the two-way capacity benchmarks quantum communication assisted by classical communication (see Sec.V in~\cite{BennettBrassardEA96}).
While there are well-known expressions for the one-way quantum capacity in
terms of a channel's coherent information~\cite{Lloyd97, Shor02a, Devetak05, DevetakWinter_2005}, no comparable expressions are known
for the two-way quantum capacity. As a result the field of quantum Shannon
theory for quantum communication remains open.

The absence of a capacity expression makes it challenging to find lower bounds
on the two-way quantum capacity of general channels. Prior work on such bounds
focuses instead on specific states and channels of interest to gain insight.

The earliest asymptotic protocols for entanglement distillation are breeding and
hashing~\cite{BennettBrassardEA96}. They apply to a
mixture of Bell states and proceed by doing random parity checks. The authors of \cite{VV05} elevate
the breeding protocol by proposing to do an asymptotic version of a single
parity check prior to carrying out breeding. These protocols give lower bounds
on $\mathcal{Q}_2$ for qubit Pauli channels.
For a general channel, the Devetak-Winter~\cite{DevetakWinter_2005} protocol achieves a rate for one-way
entanglement distillation equaling the channel's coherent information. Finding
a truly two-way protocol to exceed this rate can be hard. However, simply reversing the role of channel
sender and receiver in the Devetak-Winter~\cite{DevetakWinter_2005} scheme gives the best known bound for
the qubit amplitude damping channel called the channel's reverse coherent
information~\cite{HorodeckiHorodeckiEA00a, DevetakWinter_2005, DevetakJungeEA06, GarciaPatronPirandolaEA09}.

Directly using the breeding or the \cite{VV05} scheme does not necessarily result in the
best known rates for two-way entanglement distillation. When distilling a mixture of Bell states, a
finite-sized protocol prior to an asymptotic protocol can lead to non-trivial
improvements in the rate~(\cite{BDSW96,VV05}). A finite-sized protocol can select for
matching a $ZZ$ or
$XX$ parity check on two noisy two-qubits states~(sometimes called the recurrence protocol), however it can be useful to do
several parity checks at once to localize errors across several two-qubit states~\cite{leung2007adaptive}. Instead of
localizing errors on states, Hostens et.~al~\cite{AsympAdaptive06} come up with a sequence of parity
checks aimed at reducing the entropy of the Pauli error distribution.

Progress on finding optimal rates for distillation is crucial.  Unfortunately,
progress has been slow, for instance rates for the amplitude damping channel
haven't been improved in a decade and progress on Pauli channels has been far
and few between. 
Combining finite and asymptotic protocols for distillation can be crucial to
making progress, however the interplay between finite-sized protocols and
asymptotic protocols remains poorly understood.
For instance, the earliest protocols combining finite-sized recurrence with
hashing do not adapt the finite-sized protocol, the Vollbrecht \&
Verstraete~\cite{VV05}, as well as the Hostens-Dahene-De Moor~\cite{AsympAdaptive06} protocols, both require some
form of recurrence-based preprocessing to improve on hashing, but the best way to select
a finite-sized protocol isn't necessarily clear. 
Given that there are a variety of finite-sized protocols, each developed from a
different perspective, there is a need to investigate how these approaches
relate to one another and feed into asymptotic schemes for distillation.

Our main contribution is to achieve state-of-the-art rates for (two-way classically
assisted)~entanglement distillation across the qubit Pauli and amplitude
damping channels in the high noise regime. Key to obtaining these rates is a
conceptual insight: any finite-sized entanglement distillation protocol is a
way to reshape one noisy channel into a better channel. This reshaping idea captures the use of error-correcting codes and
two-way classical communication to distill entanglement but more importantly it
is amenable to concatenation with itself and gives intuition to help combine
different asymptotic and finite-sized distillation protocols with each other. 

We reshape the amplitude damping~(AD) channel using constant weight encodings
to an improved channel and obtain rates that exceed the channel's reverse
coherent information. 
Our simplest dual-rail and triple-rail encodings reshape the damping channel
into an erasure channel.  Using two-way capacity achieving schemes for the
erasure channel we obtain rates that exceed the AD channel's reverse coherent
information. In certain noise regimes, more complex encodings achieving even
higher rates use Hamming-weight-two encodings to reshape the AD channel into a
direct sum of erasure and non-erasure channels. Entanglement from non-erasure
channels can is extracted using different protocols~(see discussion in Sec. \ref{ss:hw2}).

Similarly to the AD case, the idea of reshaping turns out to be beneficial for
thinking about distillation over Pauli channels. We focus on the high noise
regime where rounds of pre-processing using recurrence are typically necessary
for achieving non-zero rates.  We observe that several rounds of successful
recurrence across a Pauli channel give a new improved Pauli noise channel. For
this new channel, a higher rate parity check is more suitable for detecting
errors. Thus, we concatenate a number of rounds of recurrence with a round of
parity checks of the $\llbracket 4,2,2\rrbracket$ code. Moreover, we note how each recurrence
round reshapes two channel uses into a new qubit Pauli channel, where the
dominant Pauli noise component changes according to the check performed. Thus,
we design a heuristic Greedy recurrence protocol to adapt to this change in the
noise distribution, and prove that this heuristic is guaranteed to improve the
entanglement fidelity at every round. Using these ideas, we achieve
state-of-the-art rates for entanglement distillation over the depolarizing
channel.

\subsection{Notation \& Preliminaries}
\label{sec:Not1}
We denote the entropy of a discrete probability distribution $[p_i]$ by
$h([p_i])\coloneqq-\sum_i p_i \log_2 p_i$, the binary entropy function as
$h_b(x) \coloneqq -x \log_2 x-(1-x)\log_2(1-x),$ and the Von Neumann entropy by
$S(\rho) \coloneqq -\Tr[\rho \log_2 \rho].$ A maximally entangled state of
dimension $d$ is denoted as 
$\ket{\Phi_d} \coloneqq \sum_{i=0}^{d-1}\ket{ii}/\sqrt{d}.$ The Pauli $Z$
matrix is defined as~$Z \coloneqq \ketbra{0} - \ketbra{1}$, while $Z_i\coloneqq
\mathbbm{I}^{\otimes i-1}\otimes Z\otimes \mathbbm{I}^{\otimes n-i}$ denotes Pauli~$Z$ applied to
the $i^{th}$ qubit. Similarly, the Pauli $X$ matrix is defined as $X \coloneqq
\ketbra{+} - \ketbra{-}$, while $X_i\coloneqq \mathbbm{I}^{\otimes i-1}\otimes X\otimes
\mathbbm{I}^{\otimes n-i}$, where $\ket{\pm} = (\ket{0} \pm \ket{1}) \sqrt{2}$. We denote the outcomes of a $Z$ measurement to be $+1$ and $-1$ corresponding to eigenstates $\ket{0}$ and $\ket{1}$, respectively. This same $\pm 1$ convention is used to denote outcomes of measuring products of Pauli operators $Z_i$ or $X_i$. Such product measurements are sometimes called parity measurements.

A quantum channel is a completely positive, trace-preserving map; it maps valid
quantum states to valid quantum states. A particularly useful representation
of quantum channels is the Kraus representation, where the action of a channel,
say $\mathcal{E}$, is represented via its Kraus operators $\{K_i\}$ as
$\mathcal{E}(\rho) = \sum_i K_i \rho K_i^\dagger,$ where $\sum_i K_i^\dagger
K_i=\mathbbm{I}$, the identity operator and $\rho$
represents a density operator~(see Ch.~4 in~\cite{Wilde17a} or Ch.2 in~\cite{Watrous18} for more details).

The coherent information of a bipartite state $\rho_{AB}$ is given by $
I_c(\rho_{AB})\coloneqq S(\rho_B)-S(\rho_{AB}), $ while the coherent
information of a quantum channel $\mathcal{E}_{A^\prime \rightarrow B}$ is
given by $ I_c(\mathcal{E})\coloneqq
\max_{\sigma_{AA^\prime}}I_c(\mathcal{I}_A\otimes\mathcal{E}_{A^\prime
\rightarrow B}(\sigma_{AA^\prime}))~$~\cite{SchumacherNielsen96, Lloyd97,
BarnumKnillEA00}. Meanwhile, the reverse coherent information of a bipartite
state $\rho_{AB}$ is given by $ I_r(\rho_{AB})\coloneqq S(\rho_A)-S(\rho_{AB}),
$ while the reverse coherent information of a quantum channel
$\mathcal{E}_{A^\prime \rightarrow B}$ is given by $ I_r(\mathcal{E})\coloneqq
\max_{\sigma_{AA^\prime}}I_r(\mathcal{I}_A\otimes\mathcal{E}_{A^\prime
\rightarrow B}(\sigma_{AA^\prime}))$~\cite{HorodeckiHorodeckiEA00a,
DevetakWinter_2005, DevetakJungeEA06, GarciaPatronPirandolaEA09}.
%%% NEW
We will be studying the problem of entanglement distillation (sometimes referred to as entanglement sharing or entanglement generation or entanglement purification). 
Entanglement distillation may refer to the setting where two parties share $n$ copies of a noisy entangled state and attempt to distill noiseless entanglement out of these copies, or to the setting where the communicating parties are connected by a noisy quantum channel that they use $n$ times with the goal of sharing noiseless entanglement at the conclusion of the protocol.
In this work, we focus on the latter meaning, though the two interpretations somewhat overlap. 
In particular, following \cite{Devetak05}, in the entanglement distillation protocols we consider, Alice and Bob employ $n$ uses of a noisy quantum channel with the goal of ending up with a shared state $\rho$ such that $\bra{\Phi_{d_n}}\rho\ket{\Phi_{d_n}} \geq 1-\epsilon_n$, and $\epsilon_n \rightarrow 0$, as $n\rightarrow \infty$. For such protocols, the distillation rate is defined as $R = \lim_{n\rightarrow\infty}\log_2(d_n)/n.$ 
%%%%

\section{Amplitude Damping Channel}
\label{sec:AmpDamp}
The amplitude damping (AD) channel is one of the most well-studied examples of non-Pauli channels. It is a well-motivated and realistic noise model, representing the decay of a quantum system from an excited state~$\ket{1}$ to low-energy (ground) state~$\ket{0}.$
We denote the amplitude damping channel with damping parameter $\gamma$ by $\mathcal{A}_{\gamma}$.
The Kraus operators of  $\mathcal{A}_{\gamma}$ are given by: 
\begin{equation}
    A_0 = \begin{bmatrix}
        1&0\\
        0&\sqrt{1-\gamma}
    \end{bmatrix}, \quad \text{and} \quad A_1 = \begin{bmatrix}
        0&\sqrt{\gamma}\\
        0&0
    \end{bmatrix}.\label{eqn:amp_damp_kraus}
\end{equation}

In the unassisted quantum communication setting, the amplitude damping channel
is one of a few rare examples where the capacity is exactly known. The
amplitude damping channel is degradable when $0 \leq \gamma < 1/2$ and
anti-degradable otherwise \cite{PhysRevA.75.012303, DevetakShor05}. Thus, its unassisted
capacity is equal to its coherent information for $0 \leq \gamma < 1/2$ and 0
otherwise.

In the two-way assisted quantum communication setting over the amplitude damping channel, the optimal rates are not known~(see~\cite{KhatriSharmaEA20} for bounds on the capacity).
The reverse coherent information has long been considered to be the best known lower bound on the two-way assisted quantum capacity of the amplitude damping channel \cite{GarciaPatronPirandolaEA09}. Using the concavity of the reverse coherent information in terms of the input state, the authors of \cite{GarciaPatronPirandolaEA09} obtain the following expression for the reverse coherent information of the amplitude damping channel
\begin{equation}\label{eqn:RCI_amp_damp}
    I_r(\mathcal{A}_{\gamma}) = \max_{0 \leq x \leq 1} h_b(x)-h_b(\gamma x).
\end{equation}

\subsection{Dual-Rail Encoding}
The dual-rail encoding has been suggested as a practical way of converting amplitude damping errors into erasure errors~\cite{Duan_2010,kubica2022erasurequbitsovercomingt1}, which are in turn easier to handle and correct.
In this encoding, an input logical qubit state is encoded into two physical qubits, according to the following mapping $$
\ket{0} \rightarrow \ket{01}, \quad \text{and} \quad
\ket{1} \rightarrow \ket{10}.
$$
Considering all the possible amplitude damping errors applied to the two physical qubits, we have 
\begin{align*}
    A_0 \otimes A_0 \ket{01} &= \sqrt{1-\gamma}\ket{01},  A_0 \otimes A_0 \ket{10} = \sqrt{1-\gamma}\ket{10},\\
    A_0 \otimes A_1 \ket{01} &= \sqrt{\gamma}\ket{00},  A_1 \otimes A_0 \ket{10} = \sqrt{\gamma}\ket{00}.
\end{align*}
Note that $A_1\otimes A_0\ket{01} = A_1\otimes A_1\ket{01}= A_0\otimes
A_1\ket{10}= A_1\otimes A_1\ket{10}=\mathbf{0}$, i.e., 
the damping operator $A_1$ annihilates $\ket{0}$ to the zero vector $\mathbf{0}$ as $\ket{0}$ is already in the lowest
level.
Thus, the input encoded state is either mapped by the amplitude damping channel back to the codespace, or to an orthogonal subspace spanned by $\ket{00}$, which plays the role of the erasure state, as it is orthogonal to the codespace. The output state after encoding and transmission through the amplitude damping channel is: 
$$ \mathcal{A}_{\gamma}^{\otimes 2}(\rho_{\textnormal{encoded}}) = (1-\gamma)\rho_{\textnormal{encoded}}+\gamma \ketbra{00}.$$
By using code concatenation along with the dual-rail encoding,  \cite{Duan_2010} construct multi-error correcting codes.

In the two-way assisted quantum communication setting, we now show that it suffices to use the dual-rail encoding for erasure conversion or equivalently, for amplitude damping error detection rather than correction to achieve rates higher than the reverse coherent information for some parameters of $\gamma.$

\paragraph{Protocol:} Alice encodes half of an entangled Bell pair into the dual-rail encoding, and sends the encoded qubits to Bob across the amplitude damping channel, employing two channel uses per entangled pair.  
The overall encoded state is 
$$\ket{\Phi_2}_{\textnormal{encoded}}=\frac{\ket{0}_{A}\ket{01}_{A^\prime}+\ket{1}_{A}\ket{10}_{A^\prime}}{\sqrt{2}},$$ where Alice keeps the system $A$ and transmits the systems $A^\prime$ over the channel to Bob. Simply by measuring the parity $Z_{B_1}\otimes Z_{B_2}$ of the 2 qubits he receives, Bob can determine if the received state is intact or has been subjected to a damping error mapping it to the erasure state $\ket{00}.$ The rest of the protocol proceeds similarly to the simple protocol for two-way assisted quantum communication over the quantum erasure channel. If Bob detects the erasure state $\ket{00},$ he uses classical communication to tell Alice that the shared encoded entangled pair should be discarded. Otherwise, the shared entangled pair is kept and is error-free. The shared encoded pair is kept with probability $1-\gamma$. The dual-rail encoding has a rate of $1/2$, thus the overall yield of this protocol is $$Y_{\textnormal{dual-rail}}(\gamma)=\frac{1-\gamma}{2}.$$

Take $\gamma=2/3$, then, $Y_{\textnormal{dual-rail}}(2/3)=1/6>0.16666$, while $I_r(A_{\gamma=2/3}) < 0.16148.$\\
\begin{figure}
    \centering
    \input{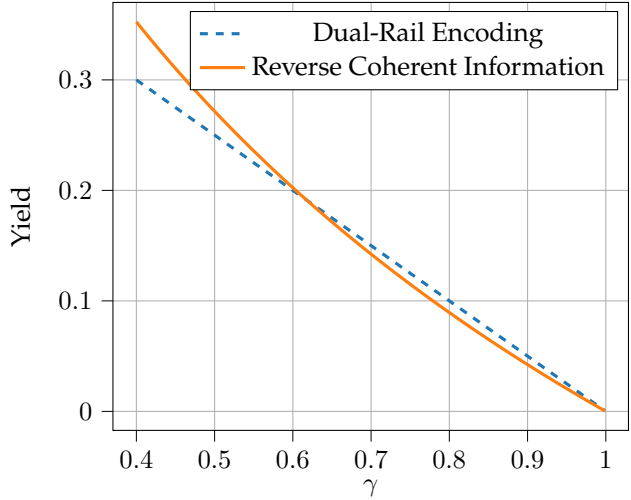}
    \caption{Comparison between the yield of the dual-rail encoding scheme and the reverse coherent information showing a range where the dual-rail encoding outperforms the reverse coherent information.}
    \label{fig:dual-rail}
\end{figure}

It is shown in \cite{siddhu2024entanglementsharingdampingdephasingchannel} that using the dual-rail encoding, the two-way assisted entanglement sharing rate can exceed that achievable by the reverse coherent information for the damping-dephasing channel. This improvement can be explained by the fact that this encoding can be used to isolate away the damping noise in the damping-dephasing channel.

\subsection{Extension to Triple-Rail Encoding} 

A generalization of the dual-rail encoding for error correction suggested in~\cite{Duan_2010}
can be adapted for improving asymptotic rates for two-way assisted entanglement distillation.
Consider an input logical qutrit state, with basis states $\{\ket{0}_L,\ket{1}_L,\ket{2}_L\},$ encoded into three physical qubits according to the mapping $$\ket{0}_L\rightarrow \ket{001}, \ket{1}_L\rightarrow \ket{010}, \ket{2}_L\rightarrow \ket{100}.$$
Using a similar protocol to the dual-rail case, Alice and Bob can distill $(\log_2 3)/3$ per channel use with probability $1-\gamma$, giving a strictly larger yield than the dual-rail case, 
\begin{equation}
Y_{\textnormal{triple-rail}}(\gamma)=\frac{(1-\gamma)\log_2 3}{3}> Y_{\textnormal{dual-rail}}(\gamma)    
\label{eq:tripleDual}
\end{equation}
\paragraph{Protocol:} Alice encodes one half of a three-dimensional entangled state $\ket{\Phi_3}$, obtaining the following encoded state  
$$\ket{\Phi_3}_{\textnormal{encoded}} = \frac{\ket{0}_A\ket{001}_{A^\prime}+\ket{1}_A\ket{010}_{A^\prime}+\ket{2}_A\ket{100}_{A^\prime}}{\sqrt{3}}.$$
Alice sends the $A^\prime$ qubits to Bob over three uses of the amplitude damping channel, with the resulting overall state being $$
\mathcal{I}_A\otimes\mathcal{A}_{\gamma}^{\otimes 3}({\Phi_3}_{\textnormal{encoded}}) = (1-\gamma)\ketbra{\Phi_3}_{\textnormal{encoded,}AB}+\gamma \frac{\mathbbm{I}_A}{3}\otimes\ketbra{000}_B.$$

Bob can then detect whether the three received qubits are in the erasure state $\ket{000}$ or in the code space by measuring the parity of the received qubits. They can do so by measuring $Z_{B_1}Z_{B_2}Z_{B_3}$ and only accepting if the outcome of the measurement is $-1$. 
Overall, with probability $(1-\gamma) $ Alice an Bob can distill $\log_2 3$ ebits over three channel uses.

\subsection{Further Hamming-Weight-1 Encodings}
One may wonder if extending this encoding to $n$ channel uses may yield even better rates. 
Consider the following encoding of half of an $n-$dimensional entangled state:
$$\ket{\Phi_n} = \frac{1}{\sqrt{n}} \sum_{i=0}^{n-1} \ket{i}_A\ket{i}_{A^\prime} \rightarrow \ket{\Phi_n}_{\textnormal{encoded}} = \frac{1}{\sqrt{n}} \sum_{i=0}^{n-1} \ket{i}_A\otimes X_{i}\ket{0}^{\otimes n}_{A^\prime}.$$
Let $A_{{x}}\coloneqq \bigotimes_{i=0}^{n-1} A_{{x}_i}, {x} \in\{0,1\}^n$ be the tensor product of the amplitude damping channel's Kraus operators, $\{ A_0,A_1\}$ defined in (\ref{eqn:amp_damp_kraus}), and
\begin{equation}
\mathbbm{A}_k\coloneqq A_0^{\otimes n-k-1}\otimes A_1 \otimes A_0^{\otimes k}.\label{eqn:single_damping_notation}
\end{equation}

Sending the $A^\prime$ systems over $n$ uses of the amplitude damping channel yields the following state shared between Alice and Bob: 
\begin{align*}
    \mathcal{I}_A\otimes\mathcal{A}_{\gamma}^{\otimes n}({\Phi_n,\textnormal{encoded}}) &= \frac{1}{n}\sum_{{x}\in\{0,1\}^n}\sum_{i=0}^{n-1} \sum_{j=0}^{n-1} \ket{i}\bra{j}_A \otimes A_{{x}} X_i\ketbra{0}^{\otimes n}_{B}X_jA_{{x}}^\dagger\\&=
\frac{1}{n}\sum_{i=0}^{n-1} \sum_{j=0}^{n-1} \ket{i}\bra{j}_A \otimes A_{0^n} X_i\ketbra{0}^{\otimes n}_{B}X_jA_{0^n}^\dagger+\frac{1}{n}\sum_{i=0}^{n-1} \ket{i}\bra{i}_A \otimes \mathbbm{A}_i X_i\ketbra{0}^{\otimes n}_{B}X_i\mathbbm{A}_i ^\dagger\\&=
(1-\gamma)\ketbra{\Phi_n}_{\textnormal{encoded}}+\gamma \frac{\mathbbm{I}_A}{n} \otimes \ketbra{0}^{\otimes n}_B
\end{align*}
Thus, by measuring the parity $Z^{\otimes n}$, Bob can determine if the received state has been `erased` into the state $\ket{0}^{\otimes n}$ or not. This way, by using the amplitude damping channel $n$ times, Alice and Bob can distill $\log_2 n$ shared ebits with probability $1-\gamma$, thus giving a scheme with yield 
\begin{equation}
Y_{\textnormal{Hamming}-1}(n,\gamma)=\frac{(1-\gamma)\log_2 n}{n}.    
\end{equation}
\begin{lemma}
\label{lem:logn_n}
   Let $n>0$ be an integer, then $(\log_2 n)/n$ is maximized at $n=3.$
\end{lemma}
\begin{proof}
    At $n=1$, $\log_2 1=0$, at $n=2$, $(\log_2 2)/2=1/2$,
    at $n=3$, $(\log_2 3)/3>0.528$.\\
For $n\geq 3$, $$\frac{\log_2 (n+1)}{n+1} - \frac{\log_2 n}{n} = \frac{\log_2\left[(1+\frac{1}{n})^n\right]-\log_2 n}{n(n+1)}\leq \frac{\log_2 \frac{e}{n}}{n(n+1)}<0.$$

\end{proof}
Thus, using Lemma \ref{lem:logn_n}, we see that the Hamming-weight-1 encodings already achieve their maximum yield at $n=3$ with the triple-rail encoding.

\subsection{Hamming-Weight-2 Encodings}\label{ss:hw2}
We now consider the case of encodings of Hamming weight-$2$ over $n$ qubits. 
For simplicity, we define the following notation: 
\begin{equation}
    \ket{{a},{b}} \coloneqq X_{{a}} X_{{b}} \ket{0}^{\otimes n}, \quad \text{and} \quad \ket{{a}} \coloneqq X_{{a}} \ket{0}^{\otimes n},
\end{equation}
where integers $a$ and $b$ lie between $0$ and $n-1$, inclusive.
The number of binary strings of length $n$ and Hamming weight 2 is $m = \binom{n}{2}=n(n-1)/2.$

\paragraph{Protocol:}
Alice encodes an $m-$dimensional entangled state as follows: 
$$\ket{\Phi_m} = \frac{1}{\sqrt{m}} \sum_{i=0}^{m-1} \ket{i}_A\ket{i}_{A^\prime} \rightarrow \ket{\Phi_m}_{\textnormal{encoded}} = \frac{1}{\sqrt{m}} \sum_{{a}=0}^{n-2}\sum_{{b}> {a}}^{n-1}  \ket{{a},{b}}_A\otimes \ket{{a},{b}}_{A^\prime}.$$

Alice then sends the $A^\prime$ systems to Bob over $n$ uses of the amplitude damping channel. Note that of the Kraus operators of $n$ applications of the amplitude damping channel to the encoded state, only those with at most two occurrences of the $A_1$ operator survive, due to the fact that $A_1\ket{0}=\mathbf{0}$.
Using the  notation from equation (\ref{eqn:single_damping_notation}), the resulting shared state can be expressed as:  

\begin{align*}
    \mathcal{I}_A\otimes\mathcal{A}_{\gamma}^{\otimes n}(\Phi_{m,\textnormal{encoded}})& = (1-\gamma)^2 \Phi_{m,\textnormal{encoded}}+\gamma^2 \frac{1}{m}\sum_{\substack{{a},\\{b}>{a}}}\ket{{a},{b}}\bra{{a},{b}}\otimes\ketbra{0}^{\otimes n}\\+&\underset{\rho_1}{\underbrace{
    \frac{1}{m}\sum_{\substack{{a},\\{b}>{a}}}\sum_{\substack{{a}^\prime,\\{b}^\prime>{a}^\prime}}\sum_{k=0}^{n-1}\ket{{a},{b}}\bra{{a}^\prime,{b}^\prime} \otimes \mathbbm{A}_k \ket{{a},{b}}\bra{{a}^\prime,{b}^\prime} \mathbbm{A}_k^\dagger}}.
\end{align*}
Focusing on the state $\rho_1$ that results from a single $A_1$ damping event, and noting that $\mathbbm{A}_k \ket{{a},{b}} = \mathbf{0}$ unless $k ={a}$ or $k ={b}$, we get:  

\begin{align*}
\rho_1 =& \frac{\gamma (1-\gamma)}{m}\biggl[\sum_{{k},{b}>{k},{b}^\prime>{k}}\ket{{{k}},{b}}\bra{{k},{b}^\prime} \otimes  \ket{{b}}\bra{{b}^\prime}
+\sum_{{k},{a}<{k},{a}^\prime<{k}}\ket{{a},{k}}\bra{{a}^\prime,{k}} \otimes  \ket{{a}}\bra{{a}^\prime}
\\&+\sum_{{k},{b}>{k},{a}^\prime<{k}}\ket{{k},{b}}\bra{{a}^\prime,{k}} \otimes  \ket{{b}}\bra{{a}^\prime}
+\sum_{{k},{a}<{k},{b}^\prime>{k}}\ket{{a},{k}}\bra{{k},{b}^\prime} \otimes  \ket{{a}}\bra{{b}^\prime}
\biggr]\\
=&\frac{\gamma (1-\gamma)}{m}\biggl[\sum_{{k},{b}\neq {k},{b}^\prime>{k}}\ket{{k},{b}}\bra{{k},{b}^\prime} \otimes  \ket{{b}}\bra{{b}^\prime}
+\sum_{{k},{a}\neq {k},{a}^\prime<{k}}\ket{{a},{k}}\bra{{a}^\prime,{k}} \otimes  \ket{{a}}\bra{{a}^\prime}\biggr]
\\
=&\frac{\gamma (1-\gamma)}{m}\biggl[\sum_{{k},{b}\neq {k},{b}^\prime\neq {k}}\ket{{k},{b}}\bra{{k},{b}^\prime} \otimes  \ket{{b}}\bra{{b}^\prime}
\biggr]\\&=2\gamma (1-\gamma)\frac{1}{n}\sum_{k}\ketbra{v_{k}},
\end{align*}
where $\ket{v_{k}} = (n-1)^{-1/2}\sum_{{b}=0:{b}\neq {{k}}}^{n-1} \ket{{b},{k}}_A\ket{{b}}_B$.
The overall state is then: 
\begin{align}
\begin{aligned}
    \mathcal{I}_A\otimes\mathcal{A}_{\gamma}^{\otimes n}(\Phi_{m,\textnormal{encoded}})& = (1-\gamma)^2 \Phi_{m,\textnormal{encoded}}+\gamma^2 \frac{1}{m}\sum_{\substack{{a},\\{b}>{a}}}\ket{{a},{b}}\bra{{a},{b}}\otimes\ketbra{0}^{\otimes n}+2\gamma (1-\gamma)\frac{1}{n}\sum_{k}\ketbra{v_k}.
    \end{aligned}
\end{align}
Bob can distinguish the case of an even number (0 or 2) of $A_1$ errors from the case of a single $A_1$ error by measuring the check $Z_{B}^{\otimes n}$ on their qubits. If the measurement outcome is $-1$, then a single error has occurred. The outcome $-1$ is obtained with probability $2 \gamma (1-\gamma).$ If the outcome is $1$, Bob can further distinguish the case of total erasure (two $A_1$ errors) from the case of no $A_1$ errors by applying the projective measurement with the projectors $\{\ketbra{0}^n,\mathbbm{I}-\ketbra{0}^{\otimes n}\}.$ 
In the case of no $A_1$ errors, $\log_2 m$ shared ebits can be distilled, as the original state is unaffected. This occurs with probability $(1-\gamma)^2.$
In the case of a single $A_1$ error, 
the state shared between Alice and Bob is $\sigma_{AB}= \sum_{k=0}^{n-1} \ketbra{v_k}/n,$ with $\bra{v_k}{v_{k^\prime}\rangle}=\delta_{k,k^\prime}$. Thus, it is a state with $n$ eigenvalues, each being $1/n$. The state $\sigma_A = \Tr_{A^\prime}[\Phi_{m,\textnormal{encoded}}]$ is up to an isometry $=\mathbbm{I}_m/m$, as Alice's systems did not go through the channel.
When Bob signals to Alice that the state they received suffered from a single damping error, they can then proceed to apply the protocol for entanglement distillation that achieves the reverse coherent information \cite{garcia2008reverse}. 
The reverse coherent information of this state is $$I_r(\sigma_{AB}) = S(\sigma_A)-S(\sigma_{AB})=\log_2(m)-\log_2(n) = \log_2 \frac{n-1}{2}.$$

The overall yield of the scheme is 
$$Y_{\textnormal{Hamming-2}}(n,\gamma)= \frac{1}{n}\left[(1-\gamma)^2 \log_2 \frac{n(n-1)}{2}+2\gamma(1-\gamma) \log_2 \frac{n-1}{2}\right].$$

In the case of $n=4,m=6$, we can do slightly better than to pass the state $\sigma_{AB}$ on to the reverse coherent strategy which achieves $\log_2[(n-1)/2]$. In this case, we may rewrite the vectors $\ket{v_k}$ in the computational basis as follows: 
\begin{align}
\begin{aligned}
&\ket{v_0} = \frac{1}{\sqrt{3}}\left(\ket{1100}\ket{0100}+\ket{1010}\ket{0010}+\ket{1001}\ket{0001}\right),
\\
&\ket{v_1} = \frac{1}{\sqrt{3}}\left(\ket{1100}\ket{1000}+\ket{0110}\ket{0010}+\ket{0101}\ket{0001}\right)
\\
&\ket{v_2} = \frac{1}{\sqrt{3}}\left(\ket{1010}\ket{1000}+\ket{0110}\ket{0100}+\ket{0011}\ket{0001}\right), \quad \text{and} \\
&\ket{v_3} = \frac{1}{\sqrt{3}}\left(\ket{1001}\ket{1000}+\ket{0101}\ket{0100}+\ket{0011}\ket{0010}\right).
\end{aligned}
\end{align}
Once Alice and Bob know that a single damping error event occurred, then Bob measures $Z_0\otimes Z_1$, and communicates the outcome to Alice. 
If the error commutes with $Z_0\otimes Z_1$, then Alice measures each of their first two qubits in the computational basis, obtaining a two bit binary outcome $b_0, b_1$. If $b_0\oplus b_1=0$, then Alice and Bob discard the state. Otherwise, they have gained a Bell pair. 
If the error anticommutes with $Z_0\otimes Z_1$ , then Alice measures each of their $2^{nd}$ and $3^{rd}$ qubits in the computational basis, obtaining a two bit outcome $b_0, b_1$. If $b_0\oplus b_1=0$, then Alice and Bob discard the state. Otherwise, they have gained a Bell pair. Thus, overall with probability $2/3$, Alice and Bob can distill a Bell pair, in the case of a single damping event, which gives the yield:
$$ Y_{\textnormal{Hamming-2}}^\ast(4,\gamma)= \frac{1}{4}\left[(1-\gamma)^2 \log_2 6+2\gamma(1-\gamma) \frac{2}{3}\right].$$
Although $Y_{\textnormal{Hamming-2}}^\ast(4,\gamma)>Y_{\textnormal{Hamming-2}}(4,\gamma)$, the rate of the original reverse coherent information protocol still exceeds both. It is an open question to devise similar improved protocols for $n>4$. From the plots of Figures~\ref{fig:Hamming-2},~\ref{fig:Hamming-2n6} and~\ref{fig:improvement}, we see that an advantage over the reverse coherent information yield (\ref{eqn:RCI_amp_damp}) can be obtained at $n=6$. 
\begin{figure}[H]
    \centering
    \begin{tikzpicture}
\definecolor{darkgray176}{RGB}{176,176,176}
\definecolor{darkorange25512714}{RGB}{255,127,14}
\definecolor{forestgreen4416044}{RGB}{44,160,44}
\definecolor{steelblue31119180}{RGB}{31,119,180}
\definecolor{plotblue}{RGB}{0,120,200}
\definecolor{plotgreen}{RGB}{0,155,130}
\definecolor{plotorange}{RGB}{240,120,50}
\definecolor{plotmagenta}{RGB}{240,50,120}
\definecolor{plotgray}{RGB}{128,128,128}
\definecolor{plotcyan}{RGB}{50,190,240}
\definecolor{plotred}{RGB}{205,50,15}
    \begin{axis}[
    legend pos=outer north east,
tick align=outside,
tick pos=left,
x grid style={darkgray176},
xlabel={$\gamma$},
xmajorgrids,
xtick style={color=black},
y grid style={darkgray176},
ylabel={Difference of Yield and Reverse Coherent information},
ymajorgrids,
ytick style={color=black}
]        \addplot[very thick,plotblue,loosely dotted] table [x=gamma,y = H24]{figz/diffH2.txt};
       \addplot[very thick,plotorange,dotted] table [x=gamma,y=H24ast]{figz/diffH2.txt};

       \addplot[very thick,plotgreen,dashed] table [x=gamma,y=H25]{figz/diffH2.txt};
       \addplot[very thick,red] table [x=gamma,y=H26]{figz/diffH2.txt};
       \addplot[ thick,steelblue31119180, mark=|] table [x=gamma,y=H27]{figz/diffH2.txt};
       \addplot[very thick,plotgray,loosely dashdotted] table [x=gamma,y=H28]{figz/diffH2.txt};
       \addplot[very thick,darkgray,densely dashdotted] table [x=gamma,y=TripleRail]{figz/diffH2.txt};

    \addlegendentry{$\Delta_{\text{RCI}}\left(Y_{\textnormal{Hamming-2}}(n=4)\right)$};
    \addlegendentry{$\Delta_{\text{RCI}}\left(Y_{\textnormal{Hamming-2}}^\ast(n=4)\right)$};

    \addlegendentry{$\Delta_{\text{RCI}}\left(Y_{\textnormal{Hamming-2}}(n=5)\right)$};
    \addlegendentry{$\Delta_{\text{RCI}}\left(Y_{\textnormal{Hamming-2}}(n=6)\right)$};
    \addlegendentry{$\Delta_{\text{RCI}}\left(Y_{\textnormal{Hamming-2}}(n=7)\right)$};
    \addlegendentry{$\Delta_{\text{RCI}}\left(Y_{\textnormal{Hamming-2}}(n=8)\right)$};
    \addlegendentry{$\Delta_{\text{RCI}}\left(Y_{\textnormal{Triple-Rail}}\right)$};

    \end{axis}
\end{tikzpicture}
    \caption{Comparison of the improvements of yields of various Hamming-2 encodings, and the Triple-Rail Encoding, over the Reverse Coherent Information rate. For any yield expression $Y$, let $\Delta_{\text{RCI}}(Y) = Y-I_r(\mathcal{A}_\gamma).$}
    \label{fig:Hamming-2}
\end{figure}
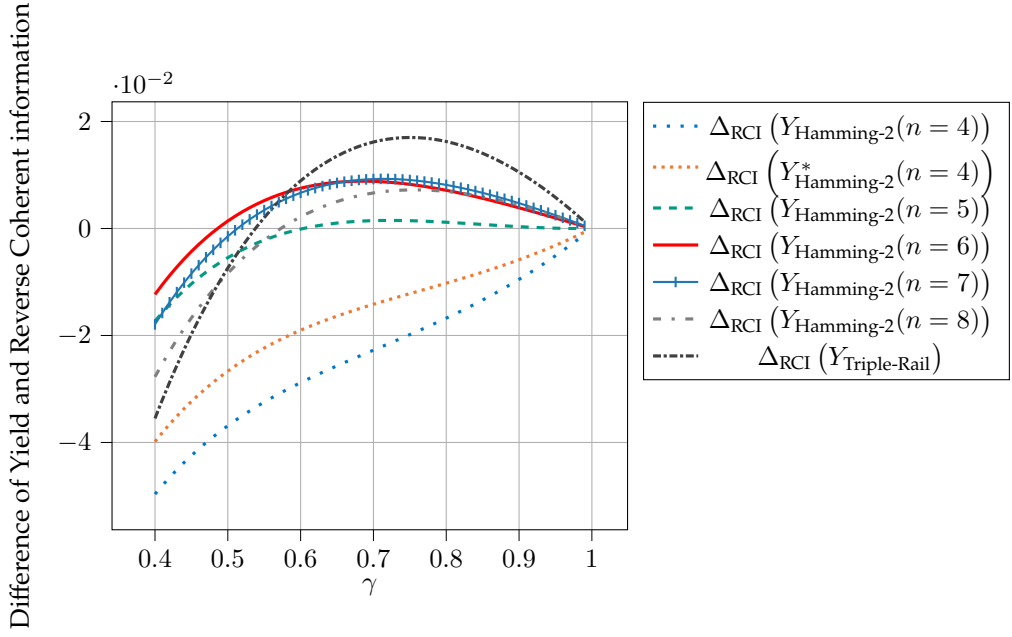

\begin{figure}[H]
    \centering
    \input{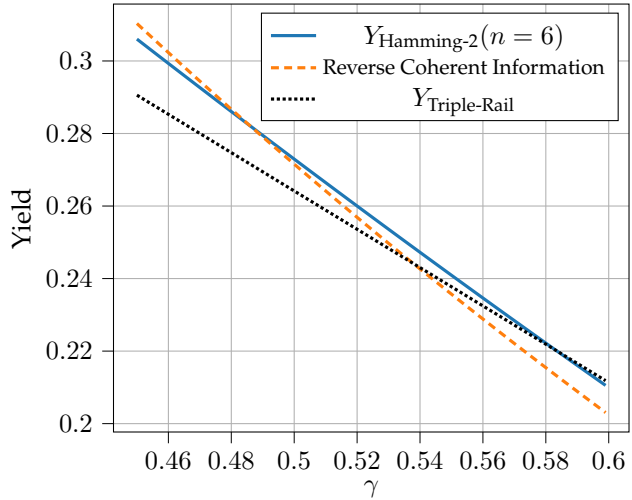}
    \caption{Comparison of the yields of the Hamming-2 encoding at $n=6$, Reverse Coherent Information, and Triple-Rail Encoding.}
    \label{fig:Hamming-2n6}
\end{figure}

\begin{figure}[H]
    \centering
    \input{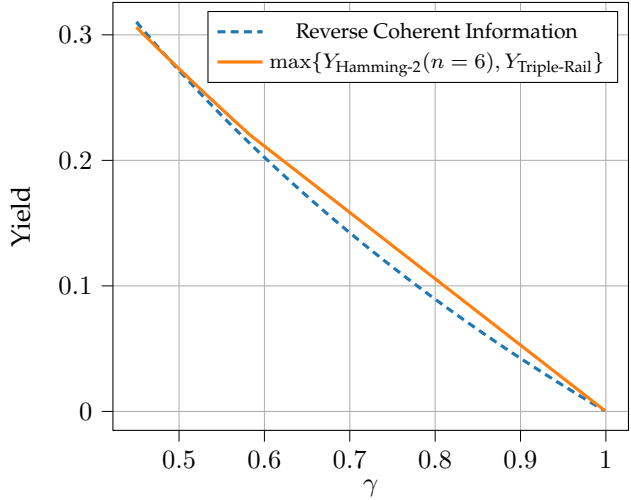}
    \caption{Comparison of the maximum of the yields of the Hamming-2 encoding at $n=6$ and the Triple-Rail Encoding vs. the Reverse Coherent Information.}
    \label{fig:improvement}
\end{figure}

\section{Pauli Channels}\label{Sec:Pauli}
In Sec.~\ref{Sec:NotationPauliDist} we define notation that we will use to discuss distillation over Pauli channels and Sec.~\ref{SS:LitRev} discusses prior work of \cite{BDSW96,VV05,leung2007adaptive}. Readers familiar with these standard discussions may skip ahead to Sec.~\ref{SS:PropProt} where we propose new protocols. 

\subsection{Notation for Entanglement Distillation over Pauli Channels}\label{Sec:NotationPauliDist}
The four Bell states are denoted as
$\ket{B_{x,z}} = I \otimes X^x Z^z(\ket{00}+\ket{11})/\sqrt{2},$
where $x,z\in\{0,1\}$. Note that $\ket{B_{00}} =\ket{\Phi_2}$.
Graphically, Bell states are denoted as in Figure \ref{fig:Bell}, with the upper wire indicating the half of the Bell pair kept by Alice, and the lower one indicating the half sent over the noisy channel to Bob.
\begin{figure}[htb!]
\centering
\begin{quantikz}
\makeebit[angle=-60]{$\ket{B_{00}}_{A,B}$} &&  \\ &\gate{x,z}   &
\end{quantikz}
\caption{$\ket{B_{x,z}}$}
\label{fig:Bell}
\end{figure}

Let the Pauli matrices be denoted by $\mathcal{P}=\{I,X,Y,Z\}.$ A capital letter subscript attached to a Pauli operator is used to indicate the subsystem of a given the Pauli operator. 

As we will be referring to stabilizer codes, we give a brief recap of stabilizer codes here. Following \cite{Gottesman97,AmbainisGottesman06}, a stabilizer code is defined via a stabilizer group $S$, an Abelian subgroup of $\mathcal{P}^{n}$, the group of $n$-qubit Pauli matrices. The codewords of such a code $\{\ket{\psi_i}\}$ are simultaneous eigenstates of the elements of $S$, with a $+1$ eigenvalue. When a codeword is sent through an $n$-qubit Pauli channel, it is affected by Pauli error $E \in \mathcal{P}^{n}$. Any such Pauli error either commutes or anticommutes with an element $S_j\in S$. Determining whether $E$ commutes or anticommutes with every generator of $S$ gives a binary string called an error syndrome. The syndrome information makes error detection, decoding and correction possible. The bits of the syndrome can be determined by measuring the stabilizer observables via projective measurements. For each such observable $S_j$, $\{\Pi_+ = (\mathbbm{I}+S_j)/2,\Pi_- = (\mathbbm{I}-S_j)/2 \}$ define elements of a projective measurement that allows determining the commutation relation of $E$ and $S_j$.

%%%%

A Pauli string of length $n$ is a tensor product of Pauli matrices acting on $n$ qubits. Upto a global phase, such a Pauli string can be represented by a binary string of length $2n$.    In a similar manner an $n$-pair Bell state is denoted by $\ket{B_s}=\ket{B_{x_1,z_1}}\otimes \ket{B_{x_2,z_2}},\dots \otimes\ket{B_{x_n,z_n}}$, where $s=(x_1,z_1),\dots,(x_n,z_n)$ is a binary string of length $2n$.
This mapping between Paulis and binary strings is called the symplectic representation. A more detailed discussion of the symplectic representation is available in Appendix \ref{app:symplectic} and a more detailed exposition can be found in \cite{Rengaswamy_2018}.
This representation establishes a direct connection between linear parity checks applied to binary strings
and stabilizer measurements. Due to this connection between stabilizer measurements and binary checks, we can sometimes denote the results of stabilizer check measurements either by binary bits $\{0,1\}$, with $0$ signalling that the measured stabilizer commutes with the error and $1$ indicating anticommutation, or equivalently by $\{+1,-1\}$ with $+1$ indicating commutation and $-1$ signalling anticommutation as mentioned in Subsection \ref{sec:Not1}.

A Pauli noise channel with probabilities $[p_I, p_X, p_Y, p_Z]$ acts on an
input state as follows: 
\begin{equation}
    \rho \rightarrow p_I \rho + p_X X\rho X+p_Y Y \rho Y+p_Z Z \rho Z.
    \label{eq:PauliChan}
\end{equation}
%
%A Pauli noise channel acts on an input state as follows: 
%
%$$\rho \rightarrow (1-p_X-p_Y-p_Z) \rho + p_X X\rho X+p_Y Y \rho Y+p_Z Z \rho
%Z.$$ 
An important example of a Pauli noise channel is the \emph{depolarizing
channel}, where $p_X=p_Y=p_Z=p/3$ and $p_I = 1-p$. An $n$-qubit quantum state
is said to be \emph{Bell-diagonal} if it can be expressed as $\rho =
\sum_{s\in\{0,1\}^{2n}}p_s \ketbra{B_s}$ for some probability distribution
$\{p_s\}.$ 
When a Pauli channel acts on each qubit of a Bell diagonal state, the resulting
state is still Bell-diagonal.

\paragraph{Useful operations on Bell pairs}
In Figures \ref{fig:bilateral-XOR} and \ref{fig:BellH}, we review some useful operations on Bell pairs.
\begin{equation*}
    \ket{B_{x_1,z_1}}\otimes \ket{B_{x_2,z_2}} \overset{\textnormal{Bilateral XOR (Fig.\ref{fig:bilateral-XOR})}}{\longrightarrow} \ket{B_{x_1,z_1\oplus z_2}}\otimes \ket{B_{x_1\oplus x_2,z_2}} 
\end{equation*}
\begin{equation*}\ket{B_{x,z}} \overset{H^{\otimes 2}\textnormal{(Fig.\ref{fig:BellH})}}{\longrightarrow}\ket{B_{z,x}} \end{equation*}
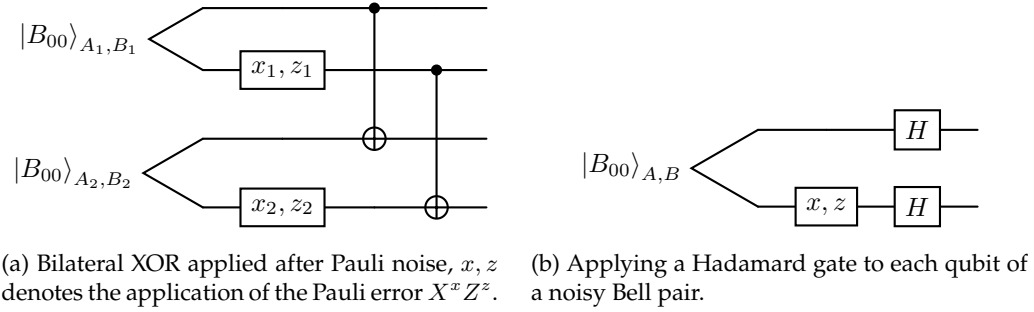
\begin{figure}[htb]
\centering
\begin{subfigure}{0.4\textwidth}
\centering
\begin{quantikz}
\makeebit[angle=-60]{$\ket{B_{00}}_{A_1,B_1}$} &              &\ctrl{2}&& \\  &\gate{x_1,z_1}  &&\ctrl{2} & \\
\makeebit[angle=-60]{$\ket{B_{00}}_{A_2,B_2}$} &&\targ{} && \\
&\gate{x_2,z_2}&&\targ{} &
\end{quantikz}
\caption{Bilateral XOR applied after Pauli noise, $x,z$ denotes the application of the Pauli error $X^xZ^z$.}
\label{fig:bilateral-XOR}
\end{subfigure}%%
\hspace{10pt}
\begin{subfigure}{0.4\textwidth}
\centering
\begin{quantikz}
\makeebit[angle=-60]{$\ket{B_{00}}_{A,B}$} &  &\gate{H}&\\ &\gate{x,z} &\gate{H}   &  
\end{quantikz}
\caption{Applying a Hadamard gate to each qubit of a noisy Bell pair.}
\label{fig:BellH}
\end{subfigure}
\caption{Useful operations on Bell pairs}
\label{fig:BellOps}
\end{figure}

\paragraph{Entanglement Fidelity}
The entanglement fidelity of a 2-qubit quantum state $\rho$ was defined in
\cite{schumacher1996sendingquantumentanglementnoisy} as, $F_e(\rho) \equiv
\bra{B_{00}}\rho\ket{B_{00}}.$ If the state is Bell-diagonal, i.e, 
\begin{equation}
    \rho = p_I \ketbra{B_{00}}+p_X \ketbra{B_{10}}+p_Y \ketbra{B_{11}}+p_Z
    \ketbra{B_{01}},
    \label{eq:1qbitBellD}
\end{equation}
then the entanglement fidelity is $p_I.$

\subsection{Literature Review}\label{SS:LitRev}

Pauli noise is one of the most widely studied noise models in quantum information, with a lot of attention being devoted to studying the depolarizing noise model. Entanglement sharing over Pauli channels has been studied in several works, with some of the main ideas established in the works of \cite{BDSW96, VV05,leung2007adaptive}. 
In all of these works, the protocols can be broadly described as follows: \begin{itemize}
    \item Alice sends  halves of Bell pairs to Bob across a noisy Pauli channel, resulting in a Bell-diagonal state shared between Alice and Bob.
    \item Alice and Bob use the resources available to them (one-way or two-way LOCC) to measure whether the Pauli error that occurred commutes or anti-commutes with a certain Pauli string or stabilizer. 
    \item Based on the outcome, Alice and Bob may decide to measure another check, or to apply a Pauli correction to their state, or to discard some qubits. 
    \item At the end of the protocol's iterations, Alice and Bob have used the channel $n$ times to end up with a shared state $\rho$, such that $\bra{B_{00}^{\otimes k}}\rho\ket{B_{00}^{\otimes k}} = 1-\epsilon_n$, with $\epsilon_n \rightarrow 0$ as $n \rightarrow \infty,$ achieving a distillation rate of $\lim_{n \rightarrow \infty}k(n)/n.$
\end{itemize}  

\paragraph{Doing parity checks.}
We consider the setting where Alice  inputs $n$ halves of Bell pairs into the Pauli noise channel, and the resulting shared state is a Bell-diagonal state of the form: 
\begin{equation}
    \rho_n = 
    \sum_{s\in\{0,1\}^{2n}}p_s \ketbra{B_s} = \sum_{P\in\mathcal{P}^n}
    \textnormal{Pr}\left[P_B\right](\mathbbm{I}_A\otimes P_B)
    \ketbra{B_{00}}^{\otimes n}(\mathbbm{I}_A\otimes P_B) .
    \label{eqn:noisyPauliState}
\end{equation}
Consider two Pauli strings $P_B,Q \in \mathcal{P}^n$, where $P_B$ is a Pauli acting on Bob's part of the state and $Q$ is a stabilizer check. Alice and Bob share a state $\ket{\psi}_{AB}=\mathbbm{I}_A\otimes P_B\ket{B_{00}}_{AB}^{\otimes n}$ and wish to determine whether $P_B$ and $Q$ commute or anti-commute. This information can be determined in two possible ways.
The first is that Alice and Bob carry out  local measurements directly on their respective parts of the shared state, consuming one of the shared noisy Bell pairs, and comparing their outcomes via classical communication to deduce the parity. 
The second is that they carry out local measurements assisted by a shared ancillary perfect Bell pair, consuming this ancillary pair in the process, then using classical communication to compute the resulting parity.  
According to the terminology used in \cite{AsympAdaptive06}, the former would be called a Bilateral Pauli Measurement (BPM), while the latter is called an Appended E-bit Measurement (AEM). For a self-contained discussion, we provide a summary regarding AEMs and BPMs in Appendix \ref{app:AEMBPM}.

\paragraph{Finite and Infinite-Size Checks. }
We are interested in optimizing the asymptotic rates achievable for entanglement distillation. In this case we assume Alice and Bob share $n$ noisy Bell pairs in the state $\sum_{s\in \{0,1\}^{2n}}p_s \ketbra{B_s}$ where $n \mapsto \infty$.
The checks they apply to their shared state can be either finite-size checks, measuring a small block size, or infinite-size checks, that is they may scale with system size.

For a finite-size parity check, Alice and Bob agree on a check with symplectic representation $r$ and a finite integer $k$, which defines the number of Bell pairs involved in a single instance of the check. In terms of the binary symplectic notation, $r$ is a binary string of length $2k$. Denote the state of $k$ of the shared noisy Bell pairs by $\sum_{t\in \{0,1\}^{2k}}p_t \ketbra{B_t}$. Alice and Bob wish to determine the value of the check given by $f_r(t) = t \cdot r \bmod 2$, for $t,r \in \{0,1\}^{2k}.$ To do so, Alice and Bob apply local Cliffords and local Pauli measurements on their respective halves of the noisy state. The measurement results they individually obtain are random, but the correlations of the measurements contain the information about the value of the check. By using classical communication to exchange the measurement outcomes, they can determine $f_r(t)$. Since we are concerned with the asymptotic setting, Alice and Bob group their infinite string of $n$ noisy shared Bell pairs into batches of finite size $k$, and determine the value of the parity check for each batch.

An infinite-size check, on the other hand, involves all of the shared pairs at once, where it determines the value of some $f_x(s) = x \cdot s \bmod 2$, for $s,x \in \{0,1\}^{2n}$ as $ n \rightarrow \infty.$

\vspace{10pt}
With these concepts in mind, we now devote the next few subsections to reviewing some of the most important protocols proposed in the literature for entanglement distillation over Pauli channels. These protocols will be particularly relevant as building blocks of the combined protocol which we will construct in following sections. For the sake of a self-contained discussion, we discuss these protocols in more detail in Appendix \ref{app:lit_review}. 

\subsubsection{Breeding \& Hashing}
Given a shared noisy state of the form (\ref{eqn:noisyPauliState}), both the
breeding and hashing protocols proposed in \cite{BDSW96} rely on carrying out
many random parity checks until the binary string $s$ can be determined. Once
$s$ is known, Bob can apply a local Pauli correction to obtain noiseless Bell
pairs.  Under the assumption that the noisy state is the result of $n$
independent and identically distributed~(i.i.d) applications of a Pauli noise
channel with probabilities $[p_I,p_X,p_Y,p_Z]$~\eqref{eq:PauliChan}, the
authors of \cite{BDSW96} use the asymptotic equipartition property to show that
$nh([p_I,p_X,p_Y,p_Z])$ parity checks would suffice to determine $s$, since
each check rules out half of the possible Pauli strings.  Thus, hashing and
breeding protocols achieve a rate
\begin{equation}
    Y_{\textnormal{Hashing}}=1-h([p_I,p_X,p_Y,p_Z]).
    \label{eq:hashing}
\end{equation}
More generally, for a Bell-diagonal state~\eqref{eqn:noisyPauliState}
the breeding/hashing {style} protocols achieve the coherent information of $\rho$,
{$n-S(\rho)=n-h(\{p_s\})$~\cite{DevetakWinter_2005}}.

The difference between breeding and hashing is that the former uses AEMs, while
the latter uses BPMs.  It is worth noting that such protocols require no
backward communication from Bob to Alice. 

\subsubsection{Recurrence}
\label{sec:recurrence}
The $Z$-recurrence protocol was introduced in \cite{BDSW96} as a preprocessing
step preceding the hashing protocol, where a BPM finite-size check is applied
in a recursive manner. The input of one step of the recurrence protocol is
pairs of noisy Bell pairs $I_{A}\otimes P_{B_1,B_2}\ket{B_{00}}_{A_1
B_1}\otimes\ket{B_{00}}_{A_2 B_2}$. All pairs affected by any $2$-qubit Pauli
error $P_{B_1,B_2}$ that anticommutes with $Z_{B_1}Z_{B_2}$, are rejected. Note
this is exactly the stabilizer of a size two $Z$-type repetition code. 
The resulting pair, if accepted, has a new induced Pauli channel acting on it.
Each step of recurrence halves the rate, as one out of every 2 Bell pairs is
consumed by the measurement. One can repeat the recurrence iterations to
further improve the induced channel, or switch over to hashing, depending on
which yields a better overall rate. Recurrence inherently requires backward
communication from Bob to Alice to communicate which Bell pairs are to be
accepted, and which are to be discarded. For the depolarizing channel combining
recurrence with hashing, the authors of \cite{BDSW96} show that two-way
communication rates are provably higher than those possible with one-way
communication only.
Analysis in~\cite{AsympAdaptive06} suggests the power of recurrence is due to
entropy reduction, where degenerate Pauli errors (i.e., an error $P_{B_1,B_2}$,
and $Z_{B_1}Z_{B_2}P_{B_1,B_2}$) are mapped together in the new induced
channel. The best rates shown in \cite{BDSW96} are achieved using the
Macchiavello recurrence protocol, discussed in~\cite[Section 3.2.1]{BDSW96}, where rounds
of $Z$-recurrence are applied, with each round being followed by a $B_x$
rotation.

\subsubsection{Interpolation of Recurrence and Hashing Protocol \cite{VV05}}
In~\cite{VV05}, the authors introduce the interpolation of recurrence and
hashing protocol, achieving rates exceeding
$Y_{\textnormal{Hashing}}$~\eqref{eq:hashing} for the depolarizing channel.
This advantage of~\cite{VV05} over hashing is observed at all depolarizing
noise rates. This protocol utilizes $2$-way classical communication between
Alice and Bob.

The protocol consists of a step of partial breeding that uses many random
infinite-size AEM checks.
The check determines if the Pauli error on two consecutive Bell pairs is
in $\mathcal{P}_{\text{even}}=\{IZ,ZI,XX,YY,ZZ,II,XY,YX\}$ or $
\mathcal{P}_{\text{odd}}=\{XZ,ZX,XI,YI,IX,IY,YZ,ZY\}$.
This step projects the state of the $i^{th}$ and $(i+1)^{th}$ noisy Bell pairs
with initial state {$\rho_2$~(setting $n=2$ in~\eqref{eqn:noisyPauliState}) into
\begin{align}
    \begin{aligned}
        \rho_{\text{even}} &= \frac{1}{p_{\text{even}}}
    \sum_{P \in \mathcal{P}_{\text{even}}}
    \textnormal{Pr}\left[P_B\right](\mathbbm{I}_A\otimes P_B) 
    \ketbra{B_{00}}^{\otimes 2} (\mathbbm{I}_A\otimes P_B), 
    \quad \text{and} \\
        \rho_{\text{odd}} &= \frac{1}{p_{\text{odd}}} 
    \sum_{P \in \mathcal{P}_{\text{odd}}}
    \textnormal{Pr}\left[P_B\right](\mathbbm{I}_A\otimes P_B) 
    \ketbra{B_{00}}^{\otimes 2} (\mathbbm{I}_A\otimes P_B),
    \end{aligned}
\end{align}
with probability  
$p_{\text{even}} = \sum_{P \in \mathcal{P}_{\text{even}}}
    \textnormal{Pr}\left[P_B\right]$ 
and
$p_{\text{odd}} = \sum_{P \in \mathcal{P}_{\text{odd}}}
    \textnormal{Pr}\left[P_B\right]$, respectively .
}
If the even output is obtained, then the state $\rho_{\text{even}}$ is passed
on to the hashing protocol. Otherwise, the BPM check $ZI$ is measured, and
either the state 
$$\rho_{\text{odd},0} = \frac{p_{ZX}+p_{IX}}{p_{0}} 
 X_B \ketbra{B_{00}}X_B+\frac{p_{ZY}+p_{IY}}{p_0}Y_B \ketbra{B_{00}}Y_B$$
is obtained with probability $p_0/p_{\text{odd}}$ and then further passed on
to the hashing protocol, where $p_0 = p_{ZX}+p_{IX}+p_{ZY}+p_{IY}$, or the state 
$$\rho_{\text{odd},1} = \frac{p_{XZ}+p_{YZ}}{p_{1}} 
 Z_B \ketbra{B_{00}}Z_B+\frac{p_{XI}+p_{YI}}{p_1}\ketbra{B_{00}}$$
is obtained with probability $p_1/p_{\text{odd}}$, where $p_1 = p_{XZ}+p_{YZ}+p_{XI}+p_{YI}$ and then further passed on
to the hashing protocol.

For a detailed discussion of the protocol steps, see Appendix \ref{app:VV}. The
overall (normalized per Bell pair) yield of the \cite{VV05} protocol is 
\begin{align}
    \text{Yield}_{VV}(\rho_2)=
p_{\text{even}}\left[1- \frac{S(\rho_{\text{even}})}{2}\right]-\frac{h_b(p_{\text{even}})}{2}+\frac{p_0}{2} \left[1-S\left(\rho_{\text{odd},0}\right)\right]+\frac{p_1}{2}  \left[1-S\left(\rho_{\text{odd},1}\right)\right].
    \label{eq:yieldVV}
\end{align}

In the case of $\rho_2$ being the output of independent and identically
distributed channel uses, where $\rho_2 = \rho^{\otimes 2}$, with
$\rho$~\eqref{eq:1qbitBellD} the result of sending one half of a Bell state
through a Pauli noise channel, the yield expression given by \eqref{eq:yieldVV}
reduces to equations (33,34) of~\cite{VV05}, where it becomes apparent that the
protocol indeed outperforms hashing~\eqref{eq:hashing}, as
\begin{equation}
\label{eq:VVbetterthanH_iid}
\text{Yield}_{VV}(\rho^{\otimes 2}) = Y_{\text{hashing}} + 
\frac{p_{\text{odd}}}{4}\left(h_b\left(\frac{p_I}{p_I+p_Z}\right)+h_b\left(\frac{p_X}{p_X+p_Y}\right)\right).
\end{equation}

The protocol of~\cite{VV05} was further analyzed and generalized into a scheme
with improved rates in~\cite{AsympAdaptive06}. It is important to note that
{to achieve non-zero rates} in the high noise regimes, both the protocols
presented in~\cite{VV05} and \cite{AsympAdaptive06} still require preprocessing
via a number of iterations of the recurrence protocol.

\subsubsection{Adaptive Entanglement Purification Protocol $\texttt{AEPP}^*(4)$}

In \cite{leung2007adaptive}, the authors propose an adaptive protocol that
relies on performing a BPM (finite-size) checks on $m$ noisy Bell pairs to
determine if the Pauli error affecting the state commutes or anticommutes with
$Z^{\otimes m}$.  In case the error commutes with $Z^{\otimes m}$, the state is
passed on to the hashing protocol. Otherwise, further shorter finite-size
checks are applied to localize the error. 

A modified version of this protocol is also suggested in the same work
\cite{leung2007adaptive}, where instead of switching directly to hashing when
the error commutes with $Z^{\otimes m}$, a check in the complementary basis is
performed first. The authors use these ideas to devise the $\texttt{AEPP}^*(4)$
protocol.  Under the assumption that the input states are i.i.d, the protocol
proceeds as follows; first a $ZZZZ$ check is measured over 4 noisy Bell pairs,
giving the syndrome bit $b_1$. If the outcome is 0 (indicating the error
commutes with $ZZZZ$), then we measure $XXXX$, giving the syndrome bit $b_2$.
While the actual measurement implemented is that of $XXX$ on the three
surviving pairs, as a result of having implemented the circuit for the $ZZZZ$
measurement, we are effectively measuring if the original error before the
$ZZZZ$ check commutes with the 4-qubit check $XXXX$.

If the $XXXX$ measurement yields the outcome $1$ (indicating the error
anticommutes with $XXXX$), the state is undistillable~(as can be checked via
the condition for Bell-diagonal states being PPT (\ref{eq:BellPPTCondition}),
discussed in \cite{divincenzo2002quantum} and restated in Appendix
\ref{app:bellppt}.) and is thus discarded.  If the $XXXX$ measurement yields
the outcome $0$, then the 2-qubit state that remains is passed on to other
protocols. {The qubits in this state have correlated errors.}

Finally, if the outcome of the  $ZZZZ$ check  is 1, instead of measuring
$XXXX$, we measure $ZZ$ on the first two qubits, giving the syndrome bit
$b_2^\prime$. If the outcome is $0$, the first pair is accepted and the third
is discarded. If the outcome is $1$, then the third pair is discarded, while
the first is kept. The accepted state is then passed on to other protocols.

\subsection{Proposed Combined Protocol}\label{SS:PropProt}
We propose a combined protocol, which at high depolarizing noise rates,
achieves higher rates than the aforementioned protocols, such as the protocol
of \cite{VV05} preceded by a number of recurrence iterations or the protocol of
\cite{leung2007adaptive}. 

The proposed protocol can be generally summarized as follows: a number of
recurrence iterations are carried out, then the accepted states are Bell-states
that are acted on by an effective channel that acts as i.i.d. uses of a Pauli
channel which is improved in comparison to the original channel acting on the
initial states. The accepted states
can then either be directly passed on to the protocol of \cite{VV05} or to a
round of the $\texttt{AEPP}^*(4)$ protocol. The accepted states of the
$\texttt{AEPP}^*(4)$ protocol are passed on to the protocol of \cite{VV05}. 

Since our protocol is composed of all these building blocks, in the following
sections, we scrutinize each component, to ensure that the decisions made at
each step achieve the largest yield possible. We study the recurrence protocol,
and introduce the Greedy recurrence protocol. Moreover, we study the case when
$\texttt{AEPP}^*(4)$ protocol outputs correlated states over 2 Bell pairs
rather than i.i.d. states. Since these correlated states are then fed into the
\cite{VV05} protocol, we discuss the conditions under which the \cite{VV05}
protocol applied to more general states still provides an advantage over
hashing.

\subsubsection{Further Analysis of Recurrence}

Generally, the recurrence parity check may be a $Z$-check, where Alice and Bob
determine if the 2-qubit Pauli error commutes with the stabilizer $Z\otimes Z$,
or it can be an $X$-check measuring the stabilizer $X\otimes X$ or a $Y$-check
measuring $Y\otimes Y$. Figure \ref{fig:Qrecurrence} shows circuits for
realizing these measurements. After applying one of theses circuits, Alice and
Bob use classical communication to compare the results of their $Z$-basis
measurements. They only accept pairs where their measurement outcomes are
equal. After one recurrence iteration, the rate is halved, since out of every
two input noisy Bell pairs, at most one Bell pair survives. One may think of
the recurrence measurement as a stabilizer code with a single stabilizer,
$n=2,k=1$. These circuits impose a certain choice of logical operators and
create an induced Pauli channel on the accepted pairs.

In what follows, we assume that the input noisy Bell pairs to the recurrence
protocol are the result of sending halves of perfect Bell pairs through
independent uses of identically distributed Pauli noise channels.  Let $Q_0
\equiv X, Q_1 \equiv Y, Q_2 \equiv Z$, with the index $i$ of $Q_i$ defined
modulo $3.$ Then, the output of $Q_i$-recurrence is accepted with probability
$p_{\text{pass}}(Q_i) = (p_I+p_{Q_i})^2+(p_{Q_{i+1}}+p_{Q_{i+2}})^2.$ The
accepted output state in each of the circuits of Figure \ref{fig:Qrecurrence}
after a round of $Q_i$-recurrence will be a Bell-diagonal state with the
following Pauli error probabilities
\begin{align}
    \left[p_I^\prime =\frac{p_I^2+p_{Q_i}^2}{p_{\text{pass}}(Q_i)},p_X^\prime=\frac{p_{Q_{i+1}}^2+p_{Q_{i+2}}^2}{p_{\text{pass}}(Q_i)},
    p^\prime_{Q_{2-(i\bmod 2)}}=\frac{2p_{I}p_{Q_i}}{p_{\text{pass}}(Q_i)},
  p^\prime_{Q_{1+(i\bmod 2)}}=\frac{2p_{Q_{i+1}}p_{Q_{i+2}}}{p_{\text{pass}}(Q_i)}.\right] \label{eqn:accepted_distribution}
\end{align}

The rejected output state for any $Q_i$-recurrence will be a Bell-diagonal state $\rho_{AB} = p_I^\ast \ketbra{B_{00}}+p_X^\ast X_B\ketbra{B_{00}}X_B+p_Y^\ast Y_B\ketbra{B_{00}}Y_B+p_Z^\ast Z_B\ketbra{B_{00}}Z_B$, with $p_I^\ast=p_X^\ast$ and $p_Y^\ast=p_Z^\ast.$ 
Let \begin{align*}
    q_1 &= (p_I^\ast + p_Z^\ast) / 2,
    q_2 = (p_X^\ast + p_Y^\ast) / 2\\
    q_3 &= (p_I^\ast - p_Z^\ast) / 2,
    q_4 = (p_X^\ast - p_Y^\ast) / 2.
\end{align*}
The partial transpose of the rejected state is given by 
\begin{equation}
    \rho_{AB}^{T_B} = 
    \begin{bmatrix} 
    q_1& 0& 0& q_4 \\ 
    0& q_2&q_3&0   \\ 
    0& q_3& q_2& 0 \\ 
    q_4& 0 &0& q_1
    \end{bmatrix}
        \label{eqn:partial_transpose}
\end{equation}
Since the rejected state has $p_I^\ast=p_X^\ast$ and $p_Y^\ast=p_Z^\ast,$ we
get $q_1=q_2, q_3=q_4,$ implying that $\rho_{AB}^{T_B} = q_1 \mathbbm{I}+q_3 X
\otimes X.$ Thus, the eigenvalues of $\rho_{AB}^{T_B}$ are
$p_I^\ast,p_I^\ast,p_Z^\ast,p_Z^\ast$. All of the eigenvalues are non-negative,
thus the rejected state $\rho_{AB}$ is PPT and therefore undistillable. This is
a useful guarantee that all states discarded by the recurrence protocol, in the
i.i.d Pauli noise case, are indeed undistillable on their own.

\begin{figure}[htb]
\centering
\begin{subfigure}{.5\textwidth}
  \centering
  \begin{quantikz}
\makeebit[angle=-60]{} &              &\ctrl{2}&& \\  &\gate{x_1,z_1}  &&\ctrl{2} & \\
\makeebit[angle=-60]{} &&\targ{} && \meter{}\\
&\gate{x_2,z_2}&&\targ{} &\meter{}
\end{quantikz}
  \caption{$Z$-Recurrence, \\state before measurement:$\ket{B_{x_1,z_1\oplus z_2}}\ket{B_{x_1\oplus x_2,z_2}}$}
  \label{fig:Z}
\end{subfigure}%
\begin{subfigure}{.5\textwidth}
  \centering
  \begin{quantikz}
\makeebit[angle=-60]{} &&\gate{H}&\ctrl{2}&& \\
&\gate{x_1,z_1} &\gate{H}&&\ctrl{2} & \\
\makeebit[angle=-60]{}&&\gate{H}&\targ{} &&\meter{} \\
&\gate{x_2,z_2}&\gate{H}&&\targ{}&\meter{} 
\end{quantikz}

  \caption{$X$-Recurrence, \\state before measurement:$\ket{B_{z_1,x_1\oplus x_2}}\ket{B_{z_1\oplus z_2,x_2}}$}
  \label{fig:X}
  
\end{subfigure}

\begin{subfigure}{.5\textwidth}
  \centering
\begin{quantikz}
\makeebit[angle=-60]{} &&\ctrl{2}&\gate{H}&\ctrl{2}&\gate{H}&&&\\
&\gate{x_1,z_1}  &&\ctrl{2} &&\gate{H}&\ctrl{2}&\gate{H}& \\
\makeebit[angle=-60]{} &&\targ{} &&\targ{} &&&& \meter{} \\
&\gate{x_2,z_2}&&\targ{} &&&\targ{} &&\meter{} 
\end{quantikz}

  \caption{$Y$-Recurrence, \\state before measurement:$\ket{B_{x_1\oplus z_2,z_1\oplus z_2}}\ket{B_{x_1\oplus x_2 \oplus z_1\oplus z_2,z_2}}$}
  \label{fig:Y}
  
\end{subfigure}
\caption{Circuits for implementing $ZZ$,$XX$,$YY$ check measurements.} 
\label{fig:Qrecurrence}
\end{figure}
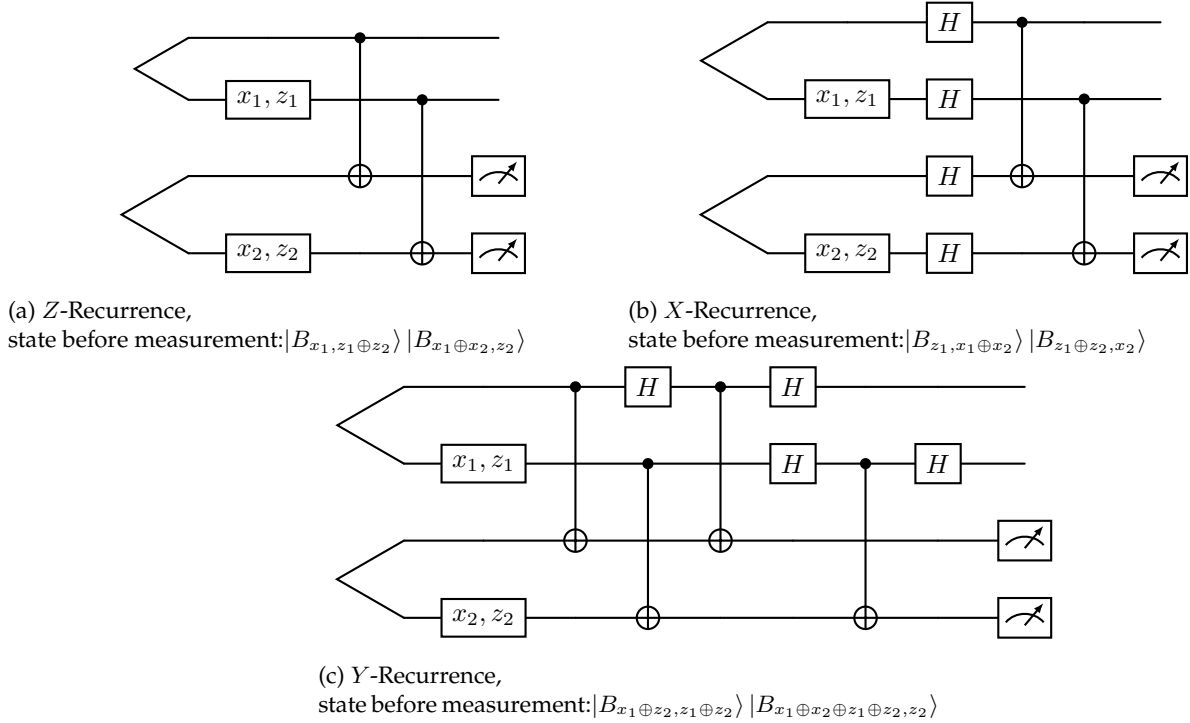
\subsubsection{Greedy Recurrence}
The essence of the recurrence protocol proposed in \cite{BDSW96} is to apply
rounds of parity checks, where the accepted states of one stage are passed on
to a recurrence check in the next stage, such that effective Pauli channel
improves gradually. However, only applying $Z$-recurrence checks will not
necessarily produce a guaranteed improvement. For example, one round of
$Z-$recurrence applied to the Pauli channel $(p_I=0.7,p_Z=0.3, p_X=p_Y=0)$
would give a channel with worse $F_e$ of $p_I^\prime =0.58.$ In \cite{BDSW96},
the authors describe a recurrence protocol with a modification suggested by C.
Macchiavello, where rounds of $Z$-recurrence are carried out, with each round
being followed by a bilateral $B_x$ rotation. Note that applying a bilateral
rotation $B_x$ followed by a round of $Z$-recurrence then applying a second
bilateral rotation $B_x$ to the accepted state results in a state with a
distribution as that produced by a round of $Y$-recurrence. This can be seen by
applying the permutation of the input Pauli noise vector according to Table
\ref{tab:bilateral_rotation} in Appendix \ref{app:Bilateral_rots} to the
probability distribution (\ref{eqn:accepted_distribution}), then applying the
same permutation to the output vector. The bilateral rotation $B_x$ exchanges
$p_I$ with $p_X$.  Thus, overall, the Macchiavello recurrence is equivalent to
interleaving rounds of $Z$ and $Y$-recurrence.

In what follows, we introduce a more structured and justified method for
choosing which type of $Q$-recurrence is to be carried out at any stage of the
recurrence procedure, called \emph{Greedy recurrence}. 

Greedy recurrence applies $Q$-recurrence whenever the input Pauli channel has
$p_Q$ as the smallest component, $Q\in\{X,Y,Z\}$. At each stage, we check if a
higher yield can be obtained by directly switching over to the protocol of
\cite{VV05} than that obtained by more rounds of recurrence first. 

The Greedy recurrence protocol is well-motivated by the following theorems that
show that it provides a guaranteed improvement of the entanglement fidelity at
each step, as well as providing the largest improvement in entanglement
fidelity among all choices of $Q$-recurrence steps.

\begin{restatable}{theorem}{guaranteedimprovement}\label{thm:guarantee_improvement}
    For any Pauli channel with Pauli error probability $[p_I,p_X,p_Y,p_Z]$,
    with $1/2<p_I<1$, the accepted state after one step of Greedy recurrence
    has a strict fidelity improvement with $p_I^\prime > p_I.$
\end{restatable}

\begin{restatable}{theorem}{optimalityofgreedyrecurrence}(Optimality of Greedy recurrence for increasing fidelity) \label{thm:optimality_of_greedy_recurrence}
Let the input to a recurrence step be the output of two i.i.d. uses of a Pauli
    noise channel to send two halves of perfect Bell pairs from Alice to Bob,
    yielding the shared state $\left(\sum_{Q\in\{I,X,Y,Z\}}p_Q Q_B
    \ketbra{B_{00}}_{AB} Q_B\right)^{\otimes 2}.$ Assume that $p_I>1/2$ and $
    p_{Q^\ast}=\min_{Q\in\{I,X,Y,Z\}} p_Q.$ Let $p_{I,Q^\ast}^\prime$ be the
    entanglement fidelity after applying $Q^\ast$-recurrence, while
    $p_{I,\Tilde{Q}}^\prime$  the entanglement fidelity after applying
    $\Tilde{Q}$-recurrence, for $\Tilde{Q} \neq Q^\ast \in\{X,Y,Z\}.$ Then,
    $p_{I,Q^\ast}^\prime \geq p_{I,\Tilde{Q}}^\prime. $    
\end{restatable}

The proofs of Theorems \ref{thm:guarantee_improvement} and
\ref{thm:optimality_of_greedy_recurrence} are provided in Appendix
\ref{app:proofs_thms_greedy}.

\paragraph{Greedy Recurrence applied to the Depolarizing Channel}
We show that Greedy recurrence and Macchiavello recurrence apply equivalent
operations (up to a final Bilateral rotation) in the case of the depolarizing
channel. 

\begin{restatable}{theorem}{MacVsGreedy}
    For the depolarizing channel with Pauli error probability
    $[1-p,p/3,p/3,p/3]$, with $0 < p < 1/2$, the sequence of alternating
    recurrence checks $Z,Y,Z,Y,\dots$ where $Z$-recurrence is performed, then
    $Y$-recurrence is performed and so on, is a valid sequence for Greedy
    recurrence. 
\end{restatable}
The proof is provided in Appendix \ref{app:proofs_thms_greedy}.

\paragraph{Greedy Recurrence applied to the $XZ$ Channel}
While Greedy recurrence reproduces the Macchiavello recurrence in the case of
the depolarizing channel, in the case of the $XZ$ channel of equation
(\ref{eq:XZchannel}), Greedy recurrence is a more useful heurisitic. 
\begin{equation}
    \rho\rightarrow (1-p)\rho+\frac{p}{2} (X\rho X+Z\rho Z).\label{eq:XZchannel}
\end{equation}
In Figure \ref{fig:MacVsGreedy}, we plot the yield obtained by applying Greedy
recurrence followed by the protocol of \cite{VV05}, and compare it to the yield
of applying the Macchiavello recurrence, again followed by the protocol of
\cite{VV05}. From the plot, it is clear that Greedy recurrence provides a more
general heuristic and better yields in the case of the $XZ$ channel.
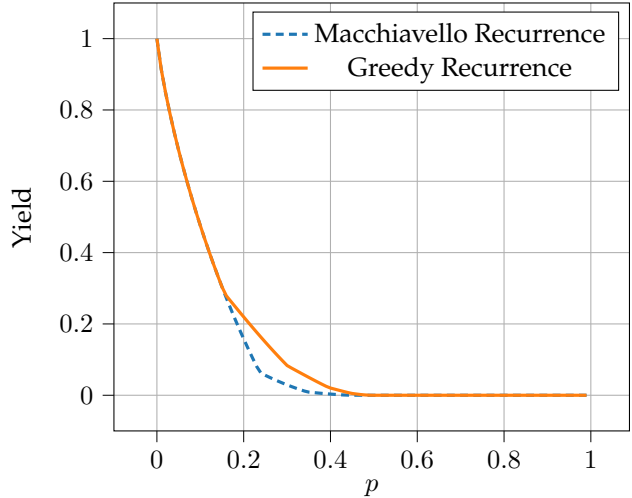
\begin{figure}[htb]
    \centering
    \begin{tikzpicture}
\definecolor{darkgray176}{RGB}{176,176,176}
\definecolor{darkorange25512714}{RGB}{255,127,14}
\definecolor{steelblue31119180}{RGB}{31,119,180}
    \begin{axis}[
tick align=outside,
tick pos=left,
x grid style={darkgray176},
xlabel={$p$},
xmajorgrids,
xtick style={color=black},
y grid style={darkgray176},
ylabel={Yield},
ymajorgrids,
ytick style={color=black}
]        \addplot[very thick,densely dashed, steelblue31119180] table [x=p,y=Macchiavello]{figz/MacVsGreedy.txt};
    
        \addplot[very thick, darkorange25512714] table [x=p,y=Greedy]{figz/MacVsGreedy.txt};
      \addlegendentry{Macchiavello Recurrence};
      \addlegendentry{Greedy Recurrence};
    \end{axis}
\end{tikzpicture}
    \caption{Comparison of the yields of the Macchiavello recurrence and Greedy
    recurrence in the case of the $XZ$ channel.}
    \label{fig:MacVsGreedy}
\end{figure}

\subsubsection{Further Analysis of the Interpolation of Recurrence and Hashing
Protocol \cite{VV05}}
It is clear from equation (\ref{eq:VVbetterthanH_iid}), that the yield of the
protocol \cite{VV05} is higher than that of hashing when the input to the
protocol is the output of independent and identically distributed uses of a
Pauli noise channel. However, in our proposed protocol, the input to the
\cite{VV05} in some cases is a correlated state of 2 shared Bell-diagonal
pairs. In the following theorem we verify that for permutation invariant states
over the qubits held by Bob of a 4 qubit shared Bell-diagonal state, such as
the output of the $\texttt{AEPP}^*(4)$ protocol when the syndrome bits
$b_1=b_2=0$, the \cite{VV05} gives an advantage over hashing.  
\begin{theorem}\label{thm:VVperm_inv}
Given a general 2-qubit Bell-diagonal state, $$\rho_2 = \sum_{Q_1,Q_2 \in \mathcal{P}^{\otimes 2}}p_{Q_1,Q_2} Q_1\otimes Q_2 \ketbra{B_{00}}^{\otimes 2} Q_1\otimes Q_2,$$ with $\{p_{Q_1,Q_2}\}$ a probability distribution over 2-qubit Paulis. If $\{p_{Q_1,Q_2}\}$ is invariant under permutation of $Q_1$ and $Q_2$, then 
$\text{Yield}_{VV}(\rho_2)\geq \text{Yield}_{\text{Hashing}}(\rho_2).$
\end{theorem}   
\begin{proof}
Using Lemma \ref{lem:entropy_of_direct_sum} twice,
    \begin{align}S(\rho_2) =& p_{\text{even}}S\left(\rho_{\text{even}}\right)+p_{\text{odd}}S\left(\rho_{\text{odd}}\right)+h_b(p_{\text{even}}) \\
    =& p_{\text{even}}S\left(\rho_{\text{even}}\right)+p_0 h\biggl(\frac{[p_{ZX},p_{IX},p_{ZY},p_{IY}]}{p_0}\biggr)
    +p_1 h\biggl(\frac{[p_{XZ},p_{XI},p_{YZ},p_{YI}]}{p_1}\biggr)\\&+p_{\text{odd}}h_b\biggl(\frac{p_0}{p_{\text{odd}}}\biggr)+h_b(p_{\text{even}})\label{eQ:Srho2}
    \end{align}
For a permutation-invariant distribution, $p_0=p_1$ and $h_b(p_0/{p_{\text{odd}}}) =1.$
Thus, the yield of the hashing protocol applied to $\rho_2$ normalized per pair is: 
\begin{align*}
    \text{Yield}_{\text{Hashing}}(\rho_2) =& 1-S(\rho_2)/2  \\=&p_{\text{even}}+p_0/2+p_1/2+p_{\text{odd}}/2 - S(\rho_2)/2\\=
    &p_{\text{even}}\left[1-S\left(\rho_{\text{even}}\right)\right]-h_b(p_{\text{even}})/2\\&+\frac{p_0}{2}\left[1- h\left(\frac{[p_{ZX},p_{IX},p_{ZY},p_{IY}]}{p_0}\right)\right]
    +\frac{p_1}{2}\left[1-h\biggl(\frac{[p_{XZ},p_{XI},p_{YZ},p_{YI}]}{p_1}\biggr)\right]
\end{align*}

    The first equality is due to the fact that $p_{\text{even}}+p_{\text{odd}} =1$ and $p_0+p_1 = p_{\text{odd}}$, while the second equality is obtained by substituting $S(\rho_2)$ using (\ref{eQ:Srho2}). 
Then, the difference of the rates achieved by the schemes is 
\begin{align*}\text{Yield}_{VV}(\rho_2)- \text{Yield}_{\text{Hashing}}(\rho_2)=& \frac{p_0}{2}\Delta_0+
\frac{p_1}{2}\Delta_1, 
\end{align*}
where \begin{align*}
\Delta_0&=h\biggl(\frac{[p_{ZX},p_{IX},p_{ZY},p_{IY}]}{p_0}\biggr)-S\left(\rho_{\text{odd},0}\right)
\\&\overset{(\dagger)}{=}
\frac{p_{ZX}+p_{IX}}{p_0}h_b\left(\frac{p_{ZX}}{p_{ZX}+p_{IX}}\right)+\frac{p_{ZY}+p_{IY}}{p_0}h_b\left(\frac{p_{ZY}}{p_{ZY}+p_{IY}}\right)
\end{align*}
\begin{align*}\Delta_1 &=h\biggl(\frac{[p_{XZ},p_{XI},p_{YZ},p_{YI}]}{p_1}\biggr)-S\left(\rho_{\text{odd},1}\right)
\\&\overset{(*)}{=}
\frac{p_{XZ}+p_{YZ}}{p_1}h_b\left(\frac{p_{XZ}}{p_{XZ}+p_{YZ}}\right)+\frac{p_{XI}+p_{YI}}{p_1}h_b\left(\frac{p_{XI}}{p_{XI}+p_{YI}}\right),
\end{align*}
where the equalities $(\dagger)$ and $(*)$ follow by using Lemma \ref{lem:entropy_reduction}.
Since $\Delta_0,\Delta_1 \geq 0$,
    $\text{Yield}_{VV}(\rho_2)\geq\text{Yield}_{\text{Hashing}}(\rho_2)$.  In
    this case, we may express $\text{Yield}_{VV}(\rho_2)$
\begin{equation*}
    \text{Yield}_{VV}(\rho_2) =
    1-S(\rho_2)/2+\frac{p_{\text{odd}}}{4}(\Delta_0+\Delta_1)
\end{equation*}
\end{proof}

\paragraph{Further optimization}
We observe that the expression in (\ref{eq:VVbetterthanH_iid}) gives different values for different permutations of the probability vector $[p_I,p_X,p_Y,p_Z]$. This can be exploited, by applying bilateral rotations and local Pauli operations to permute the vector to get the highest possible yield.

\subsubsection{Description of the Proposed Protocol and Results}
Here, we describe the steps of the proposed protocol that is built as a careful composition of the aforementioned protocols.
\begin{itemize}
    \item The input is $n$ shared noisy Bell pairs, that are generated as a
        result of sending $n$ halves of Bell pairs across $n$ i.i.d. uses of a
        Pauli channel. 
    \item Alice and Bob apply Greedy recurrence for $N_1$ iterations. At each iteration, $Q$-recurrence is applied, with $Q$ corresponding to the Pauli error with the smallest probability in the effective channel acting on the shared states. At each iteration, the rate is halved, and the possible yield is scaled by $p_{\text{pass}}/2$, while the effective Pauli channel improves, as more errors are caught. 
    \item The accepted states in the previous step are passed on to one of the two following possible branches, depending on which gives a higher yield: 
    \begin{enumerate}
        \item One choice is to apply bilateral rotations to the i.i.d. Bell-pairs, to maximize the yield, then pass on these pairs to the \cite{VV05} protocol, achieving a yield that can be computed via equation (\ref{eq:VVbetterthanH_iid}). 

        \item The other choice is to apply the  $\texttt{AEPP}^*(4)$ protocol. The accepted states from the Greedy recurrence step are grouped into groups of four, over which the check $ZZZZ$ is measured. \begin{itemize}
            \item States with errors that commute with $ZZZZ$ are subjected to
                the check $XXXX$, and if the error commutes with this check
                too, they are then passed on to the \cite{VV05} protocol. At
                this point, these are correlated states over two shared
                Bell-diagonal pairs. These correlations are useful, since the
                joint entropy of a probability distribution is smaller than the
                sum of entropies of the marginals $H(U_1,U_2)\leq
                H(U_1)+H(U_2)$.
            To take these correlations into consideration, we compute the yield of applying the \cite{VV05} protocol via the more general expression of (\ref{eq:yieldVV}). Note that the state over the 2 Bell pairs at this stage is permutation invariant with respect to Bob's qubits, and by Theorem \ref{thm:VVperm_inv}, the \cite{VV05} protocol gives a better yield than hashing in this case.
            \item States that fail the $ZZZZ$ check are subjected to a further check $ZZ$. If the check is passed, then the first Bell pair is passed on to a further $N_2$ iterations of Greedy recurrence, then to the \cite{VV05} protocol, achieving (\ref{eq:VVbetterthanH_iid}). 
            If the $ZZ$ check is failed, then the third Bell pair is passed on to a further $N_3$ iterations of Greedy recurrence, then to the \cite{VV05} protocol, achieving a yield computable by (\ref{eq:VVbetterthanH_iid}). 
        \end{itemize}
    \end{enumerate}
\end{itemize}

The numbers of iterations $N_1,N_2,N_3$ can be optimized to achieve the best
rate possible. In in Figure \ref{fig:proposed_prot}, we see an improved rate
advantage of our proposed protocol over the protocols from the literature
\cite{VV05,leung2007adaptive} for the depolarizing channel $\rho \rightarrow
(1-p)\rho+p/3(X \rho X+Y\rho Y+Z\rho Z)$. Note that the $\texttt{AEPP}(4)^\ast$
protocol already exhibits an advantage over (recurrence followed by \cite{VV05}) in the
intermediate noise regime, which can also be seen in the plots of
Ref.~\cite{leung2007adaptive}. Intuitively, this can be explained as follows, a
$Z$-recurrence step on 4 Bell pairs would catch errors such as $XIIX$, while
the first check of the $\texttt{AEPP}(4)^\ast$ protocol is $ZZZZ$ and would
commute with such an error. However, both protocols would catch any single
error on 4 pairs.  After applying the checks $ZZZZ,XXXX$ of the
$\texttt{AEPP}(4)^\ast$ protocol, the rate is halved, but single-qubit errors
in both amplitude and phase basis are caught. At the same cost of halving the
rate, the $Z$-recurrence protocols only checks for amplitude errors but manages
to catch more of them. Thus, we see that the $\texttt{AEPP}(4)^\ast$ protocol
has a better yield than recurrence in the intermediate noise regime but not in
the high noise regime.

Our proposed protocol succeeds in extending the range of this advantage of the $\texttt{AEPP}(4)^\ast$ protocol into the high noise regime, by preprocessing the highly noisy Pauli channel using recurrence, into a noise regime where the $\texttt{AEPP}(4)^\ast$ protocol has advantage once again.  
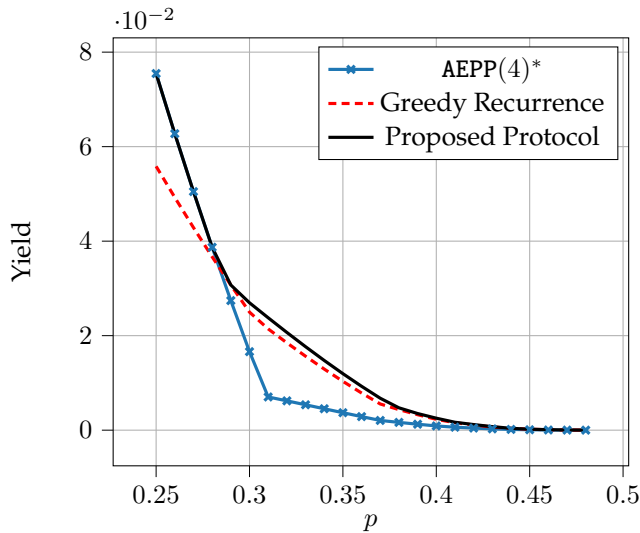
\begin{figure}[H]
    \centering
    \begin{tikzpicture}
\definecolor{darkgray176}{RGB}{176,176,176}
\definecolor{darkorange25512714}{RGB}{255,127,14}
\definecolor{steelblue31119180}{RGB}{31,119,180}
    \begin{axis}[
tick align=outside,
tick pos=left,
x grid style={darkgray176},
xlabel={$p$},
xmajorgrids,
xtick style={color=black},
y grid style={darkgray176},
ylabel={Yield},
ymajorgrids,
ytick style={color=black}
]        \addplot[very thick, steelblue31119180,mark=x] table [x=p,y=AEPP4]{figz/proposed_protocol_adv.txt};
 \addplot[very thick,densely dashed, red] table [x=p,y=Greedy]{figz/proposed_protocol_adv.txt};
  \addplot[very thick, black] table [x=p,y=ProposedProtocol]{figz/proposed_protocol_adv.txt};
      \addlegendentry{$\texttt{AEPP}(4)^*$};
        \addlegendentry{Greedy Recurrence};
              \addlegendentry{Proposed Protocol};

    \end{axis}
\end{tikzpicture}
    \caption{Comparison of the yields of Greedy recurrence+\cite{VV05}, $\texttt{AEPP}(4)^\ast$+\cite{VV05}, and the proposed protocol in the case of the depolarizing channel.}
    \label{fig:proposed_prot}
\end{figure}

\section{Discussion and Future Directions}\label{Sec:Disc}
In this work, we have provided important examples of rate improvements for
two-way assisted entanglement distillation. 
For the qubit depolarizing channel and other qubit Pauli channels,
our techniques give higher rates for entanglement distillation than the best known
well-established rates found previously. For the qubit amplitude damping channel,
our techniques gives higher rates than the channel's reverse coherent information,
a rate that has been the best known for over a decade.  
Moreover, the finite pre-processing protocols and codes we provide for the amplitude damping channel are applicable to practical non-asymptotic settings.
These results open new avenues for investigating better protocols for
distillation across a variety of channels. 

We have specifically focused on the high noise regime, where all asymptotic protocols require finite-size pre-processing parity checks to be carried out at the beginning of the protocol. Optimizing the size of the initial check used prior to an asymptotic protocol is an interesting open problem. For amplitude damping noise and Hamming-weight one codes, we have shown that the optimal check size is 3 .

Our improvements rely on shaping multiple uses of a noisy channel into an improved channel, and carefully studying how to concatenate different distillation
protocols. At each step of these protocols, we choose a distillation procedure
that is most suited to the effective noise channel shaped by the previous
stage.  Further investigation of connections between channel reshaping~(a
concept akin to channel simulation~\cite{hastings2014} but assisted by two-way
classical communication) and entanglement distillation may yield more insights. 

In the channel reshaping procedure, the effective channel can differ from the
constituent channel in two different ways. First, the reshaped channel need not
belong to the same channel family as the constituent channel~(see Sec.~\ref{sec:AmpDamp} where
several amplitude damping channels are reshaped into an erasure channel). This feature
can simplify the analysis if the properties of the reshaped channel are well known.
Second, a single use of the reshaped channel may have higher dimensions than  a single use of the constituent channel. The new noise model may result in  correlations between the constituent dimensions and thus may allow for additional (conceptual)~freedom for improving rates for entanglement distillation over the constituent channel.

One such concept is that the joint distribution in the case of correlated noise
has less entropy then the sum of entropies of the marginals. Thus, while it
might be more analytically difficult to deal with distributions that are
non-i.i.d., there might be gains from considering distillation protocols
that yield correlated output states. One direction to circumvent the difficulty of
studying general distributions might be to focus on protocols that generate
correlated states with a high degree of symmetry.

\section{Acknowledgments}

D.A. is grateful to IBM T.J. Watson Research Center's hospitality, where part of this work was carried out as part of an internship. T.J, V.S., and J.S, are supported by the U.S. Department of Energy, Office of Science, National Quantum
Information Science Research Centers, Co-design Center for Quantum Advantage (C2QA) contract (DESC0012704).

\printbibliography

@misc{hastings2014,
      title={Notes on Some Questions in Mathematical Physics and Quantum Information}, 
      author={M. B. Hastings},
      year={2014},
      eprint={1404.4327},
      archivePrefix={arXiv},
      primaryClass={quant-ph},
      url={https://arxiv.org/abs/1404.4327}, 
}

@article{DevetakWinter_2005,
   title={Distillation of secret key and entanglement from quantum states},
   volume={461},
   ISSN={1471-2946},
   url={http://dx.doi.org/10.1098/rspa.2004.1372},
   DOI={10.1098/rspa.2004.1372},
   number={2053},
   journal={Proceedings of the Royal Society A: Mathematical, Physical and Engineering Sciences},
   publisher={The Royal Society},
   author={Devetak, Igor and Winter, Andreas},
   year={2005},
pages={207–235} }

@article{PhysRevA.75.012303,
  title = {Quantum capacities of channels with small environment},
  author = {Wolf, Michael M. and P\'erez-Garc\'{\i}a, David},
  journal = {Phys. Rev. A},
  volume = {75},
  issue = {1},
  pages = {012303},
  numpages = {4},
  year = {2007},
  month = {Jan},
  publisher = {American Physical Society},
  doi = {10.1103/PhysRevA.75.012303},
  url = {https://link.aps.org/doi/10.1103/PhysRevA.75.012303}
}

@inproceedings{Duan_2010,
   title={Multi-error-correcting amplitude damping codes},
   volume={75},
   url={http://dx.doi.org/10.1109/ISIT.2010.5513648},
   DOI={10.1109/isit.2010.5513648},
   booktitle={2010 IEEE International Symposium on Information Theory},
   publisher={IEEE},
   author={Duan, Runyao and Grassl, Markus and Ji, Zhengfeng and Zeng, Bei},
   year={2010},
   month=jun, pages={2672–2676} }

@article{Bennett_1997,
   title={Capacities of Quantum Erasure Channels},
   volume={78},
   ISSN={1079-7114},
   url={http://dx.doi.org/10.1103/PhysRevLett.78.3217},
   DOI={10.1103/physrevlett.78.3217},
   number={16},
   journal={Physical Review Letters},
   publisher={American Physical Society (APS)},
   author={Bennett, Charles H. and DiVincenzo, David P. and Smolin, John A.},
   year={1997},
   month=apr, pages={3217–3220} }

@article{kubica2022erasurequbitsovercomingt1,
  title = {Erasure Qubits: Overcoming the ${T}_{1}$ Limit in Superconducting Circuits},
  author = {Kubica, Aleksander and Haim, Arbel and Vaknin, Yotam and Levine, Harry and Brand\~ao, Fernando and Retzker, Alex},
  journal = {Phys. Rev. X},
  volume = {13},
  issue = {4},
  pages = {041022},
  numpages = {20},
  year = {2023},
  month = {Nov},
  publisher = {American Physical Society},
  doi = {10.1103/PhysRevX.13.041022},
  url = {https://link.aps.org/doi/10.1103/PhysRevX.13.041022}
}

@article{BDSW96,
  title = {Mixed-state entanglement and quantum error correction},
  author = {Bennett, Charles H. and DiVincenzo, David P. and Smolin, John A. and Wootters, William K.},
  journal = {Phys. Rev. A},
  volume = {54},
  issue = {5},
  pages = {3824--3851},
  numpages = {0},
  year = {1996},
  month = {Nov},
  publisher = {American Physical Society},
  doi = {10.1103/PhysRevA.54.3824},
  url = {https://link.aps.org/doi/10.1103/PhysRevA.54.3824}
}

@article{VV05,
  title = {Interpolation of recurrence and hashing entanglement distillation protocols},
  author = {Vollbrecht, Karl Gerd H. and Verstraete, Frank},
  journal = {Phys. Rev. A},
  volume = {71},
  issue = {6},
  pages = {062325},
  numpages = {7},
  year = {2005},
  month = {Jun},
  publisher = {American Physical Society},
  doi = {10.1103/PhysRevA.71.062325},
  url = {https://link.aps.org/doi/10.1103/PhysRevA.71.062325}
}

@article{AsympAdaptive06,
  title = {Publisher's Note: Asymptotic adaptive bipartite entanglement-distillation protocol [Phys. Rev. A73, 062337 (2006)]},
  author = {Hostens, Erik and Dehaene, Jeroen and De Moor, Bart},
  journal = {Phys. Rev. A},
  volume = {74},
  issue = {1},
  pages = {019903},
  year = {2006},
  month = {Jul},
  publisher = {American Physical Society},
  doi = {10.1103/PhysRevA.74.019903},
  url = {https://link.aps.org/doi/10.1103/PhysRevA.74.019903}
}

@article{leung2007adaptive,
  title={Adaptive entanglement purification protocols with two-way classical communication},
  author={Leung, Alan W and Shor, Peter W},
  journal={arXiv preprint quant-ph/0702156},
  year={2007}
}

@misc{RicochetPty,
  author        = {Ian Lim},
  title         = {Lecture notes in Quantum Information Theory},
  month         = {June},
  year          = {2019},
url={https://lim.physics.ucdavis.edu/teaching/files/qi-notes-partiii.pdf}
}

@inproceedings{Rengaswamy_2018,
   title={Synthesis of Logical Clifford Operators via Symplectic Geometry},
   url={http://dx.doi.org/10.1109/ISIT.2018.8437652},
   DOI={10.1109/isit.2018.8437652},
   booktitle={2018 IEEE International Symposium on Information Theory (ISIT)},
   publisher={IEEE},
   author={Rengaswamy, Narayanan and Calderbank, Robert and Pfister, Henry D. and Kadhe, Swanand},
   year={2018},
   month=jun, pages={791–795} }

@misc{schumacher1996sendingquantumentanglementnoisy,
      title={Sending quantum entanglement through noisy channels}, 
      author={Benjamin Schumacher},
      year={1996},
      eprint={quant-ph/9604023},
      archivePrefix={arXiv},
      primaryClass={quant-ph},
      url={https://arxiv.org/abs/quant-ph/9604023}, 
}

@article{divincenzo2002quantum,
  title={Quantum data hiding},
  author={DiVincenzo, David P and Leung, Debbie W and Terhal, Barbara M},
  journal={IEEE Transactions on Information Theory},
  volume={48},
  number={3},
  pages={580--598},
  year={2002},
  publisher={IEEE}
}

@misc{siddhu2024entanglementsharingdampingdephasingchannel,
      title={Entanglement sharing across a damping-dephasing channel}, 
      author={Vikesh Siddhu and Dina Abdelhadi and Tomas Jochym-O'Connor and John Smolin},
      year={2024},
      eprint={2405.06231},
      archivePrefix={arXiv},
      primaryClass={quant-ph},
      url={https://arxiv.org/abs/2405.06231}, 
}

@Article{	  AmbainisGottesman06,
  title		= {The minimum distance problem for two-way entanglement
		  purification},
  volume	= {52},
  issn		= {0018-9448},
  url		= {http://ieeexplore.ieee.org/document/1580811/},
  doi		= {10.1109/TIT.2005.862089},
  number	= {2},
  urldate	= {2023-05-01},
  journal	= {IEEE Transactions on Information Theory},
  author	= {Ambainis, A. and Gottesman, D.},
  month		= feb,
  year		= {2006},
  pages		= {748--753}
}

@Article{	  BarnumKnillEA00,
  title		= {On quantum fidelities and channel capacities},
  volume	= {46},
  issn		= {00189448},
  url		= {http://ieeexplore.ieee.org/document/850671/},
  doi		= {10.1109/18.850671},
  number	= {4},
  urldate	= {2023-06-28},
  journal	= {IEEE Transactions on Information Theory},
  author	= {Barnum, H. and Knill, E. and Nielsen, M.A.},
  month		= jul,
  year		= {2000},
  pages		= {1317--1329}
}

@Article{	  BennettBrassardEA96,
  title		= {Purification of {Noisy} {Entanglement} and {Faithful}
		  {Teleportation} via {Noisy} {Channels}},
  volume	= {76},
  issn		= {0031-9007, 1079-7114},
  url		= {http://arxiv.org/abs/quant-ph/9511027},
  doi		= {10.1103/PhysRevLett.76.722},
  number	= {5},
  urldate	= {2023-05-01},
  journal	= {Physical Review Letters},
  author	= {Bennett, Charles H. and Brassard, Gilles and Popescu,
		  Sandu and Schumacher, Benjamin and Smolin, John A. and
		  Wootters, William K.},
  month		= jan,
  year		= {1996},
  note		= {arXiv:quant-ph/9511027},
  pages		= {722--725}
}

@Article{	  BennettShor98,
  author	= {C. H. Bennett and P. W. Shor},
  journal	= {IEEE Transactions on Information Theory},
  title		= {Quantum information theory},
  year		= {1998},
  volume	= {44},
  number	= {6},
  pages		= {2724-2742},
  doi		= {10.1109/18.720553},
  issn		= {0018-9448},
  month		= {Oct}
}

@Article{	  Devetak05,
  title		= {The {Private} {Classical} {Capacity} and {Quantum}
		  {Capacity} of a {Quantum} {Channel}},
  volume	= {51},
  issn		= {0018-9448},
  url		= {http://ieeexplore.ieee.org/document/1377491/},
  doi		= {10.1109/TIT.2004.839515},
  number	= {1},
  urldate	= {2023-06-28},
  journal	= {IEEE Transactions on Information Theory},
  author	= {Devetak, I.},
  month		= jan,
  year		= {2005},
  pages		= {44--55}
}

@Article{	  DevetakShor05,
  year		= {2005},
  issn		= {0010-3616},
  journal	= {Communications in Mathematical Physics},
  volume	= {256},
  number	= {2},
  doi		= {10.1007/s00220-005-1317-6},
  title		= {The Capacity of a Quantum Channel for Simultaneous
		  Transmission of Classical and Quantum Information},
  publisher	= {Springer-Verlag},
  author	= {Devetak, I. and Shor, P. W.},
  pages		= {287-303}
}

@Article{	  GarciaPatronPirandolaEA09,
  title		= {Reverse {Coherent} {Information}},
  volume	= {102},
  issn		= {0031-9007, 1079-7114},
  url		= {http://arxiv.org/abs/0808.0210},
  doi		= {10.1103/PhysRevLett.102.210501},
  number	= {21},
  urldate	= {2023-05-01},
  journal	= {Physical Review Letters},
  author	= {Garc\'{i}a-Patr\'{o}n, Ra\'{u}l and Pirandola, Stefano and Lloyd,
		  Seth and Shapiro, Jeffrey H.},
  month		= {May},
  year		= {2009},
  note		= {arXiv:0808.0210 [quant-ph]},
  pages		= {210501}
}

@Misc{		  Gottesman97,
  title		= {Stabilizer {Codes} and {Quantum} {Error} {Correction}},
  url		= {http://arxiv.org/abs/quant-ph/9705052},
  urldate	= {2023-05-01},
  publisher	= {arXiv},
  author	= {Gottesman, Daniel},
  month		= may,
  year		= {1997},
  note		= {arXiv:quant-ph/9705052}
}

@Article{	  KhatriSharmaEA20,
  title		= {Information-theoretic aspects of the generalized
		  amplitude-damping channel},
  author	= {Khatri, Sumeet and Sharma, Kunal and Wilde, Mark M.},
  journal	= {Phys. Rev. A},
  volume	= {102},
  issue		= {1},
  pages		= {012401},
  numpages	= {31},
  year		= {2020},
  month		= {Jul},
  publisher	= {American Physical Society},
  doi		= {10.1103/PhysRevA.102.012401}
}

@Article{	  KretschmannWerner04,
  title		= {Tema {Con} {Variazioni}: {Quantum} {Channel} {Capacity}},
  volume	= {6},
  issn		= {1367-2630},
  shorttitle	= {Tema {Con} {Variazioni}},
  url		= {http://arxiv.org/abs/quant-ph/0311037},
  doi		= {10.1088/1367-2630/6/1/026},
  urldate	= {2023-05-01},
  journal	= {New Journal of Physics},
  author	= {Kretschmann, Dennis and Werner, Reinhard F.},
  month		= feb,
  year		= {2004},
  note		= {arXiv:quant-ph/0311037},
  pages		= {26--26}
}

@Article{	  Lloyd97,
  title		= {Capacity of the noisy quantum channel},
  author	= {Lloyd, Seth},
  journal	= {Phys. Rev. A},
  volume	= {55},
  issue		= {3},
  pages		= {1613--1622},
  numpages	= {0},
  year		= {1997},
  month		= {Mar},
  publisher	= {American Physical Society},
  doi		= {10.1103/PhysRevA.55.1613}
}

@Article{	  SchumacherNielsen96,
  title		= {Quantum data processing and error correction},
  author	= {Schumacher, Benjamin and Nielsen, M. A.},
  journal	= {Phys. Rev. A},
  volume	= {54},
  issue		= {4},
  pages		= {2629--2635},
  numpages	= {0},
  year		= {1996},
  month		= {Oct},
  publisher	= {American Physical Society},
  doi		= {10.1103/PhysRevA.54.2629}
}

@Misc{		  Shor02a,
  author	= {Peter W. Shor},
  title		= {Quantum error correction},
  month		= {Nov},
  year		= {2002},
  publisher	= {Lecture Notes, MSRI Workshop on Quantum Computation.},
  url		= {http://www.msri.org/workshops/203/schedules/1181}
}

@Book{		  Watrous18,
  edition	= {1},
  title		= {The {Theory} of {Quantum} {Information}},
  isbn		= {978-1-316-84814-2 978-1-107-18056-7},
  url		= {https://www.cambridge.org/core/product/identifier/9781316848142/type/book},
  urldate	= {2023-05-01},
  publisher	= {Cambridge University Press},
  author	= {Watrous, John},
  month		= apr,
  year		= {2018},
  doi		= {10.1017/9781316848142}
}

@Book{		  Wilde17a,
  place		= {Cambridge},
  edition	= {2},
  title		= {Quantum Information Theory},
  doi		= {10.1017/9781316809976},
  publisher	= {Cambridge University Press},
  author	= {Wilde, Mark M.},
  year		= {2017}
}

@article{garcia2008reverse,
  title={Reverse Coherent Information},
  author={Garc{\'\i}a-Patr{\'o}n, Ra{\'u}l and Pirandola, Stefano and Lloyd, Seth and Shapiro, Jeffrey H},
  journal={arXiv preprint arXiv:0808.0210},
  year={2008}
}

@Article{	  DevetakJungeEA06,
  author	= {Devetak, Igor and Junge, Marius and King, Christoper and
		  Ruskai, Mary Beth},
  title		= {Multiplicativity of Completely Bounded p-Norms Implies a
		  New Additivity Result},
  journal	= {Communications in Mathematical Physics},
  year		= {2006},
  month		= {Aug},
  day		= {01},
  volume	= {266},
  number	= {1},
  pages		= {37-63},
  issn		= {1432-0916},
  doi		= {10.1007/s00220-006-0034-0}
}

@Article{	  HorodeckiHorodeckiEA00a,
  title		= {Unified Approach to Quantum Capacities: Towards Quantum Noisy
      Coding Theorem},
  author	= {Horodecki, Micha\l{} and Horodecki, Pawe\l{} and
		  Horodecki, Ryszard},
  journal	= {Phys. Rev. Lett.},
  volume	= {85},
  issue		= {2},
  pages		= {433--436},
  numpages	= {0},
  year		= {2000},
  month		= {Jul},
  publisher	= {American Physical Society},
  doi		= {10.1103/PhysRevLett.85.433}
}
\appendix
\section{Useful properties of Bell states}
\subsection{Ricochet property}
A useful property of Bell states is the Ricochet property (or the transpose trick) \cite{RicochetPty}.
\begin{lemma}(Ricochet Property)
For any $A\in \mathbbm{C}^{d\times d}$, and an orthonormal basis $\{\ket{x}\}$ of $\mathbbm{C}^d,$
$$A\otimes I \sum \ket{xx} = I\otimes A^T \sum \ket{xx},$$
where $T$ denotes the transpose of a matrix.
\label{lem:ricochet}
\end{lemma}

\begin{proof}
We can express the operator $A$ as $A = \sum_{u,w} a_{u,w} \ket{u}\bra{w}.$
The left hand side can now be written as  
$\sum_{u,x} a_{u,x}\ket{ux}.$

The right hand side is $
\sum_{u,w} a_{u,w}I \ket{x}\otimes \ket{w}\bra{u}\ket{x}= \sum_{x,w} a_{x,w}\ket{xw}, 
$ which equals the left hand side by renaming the variables.
\end{proof}
\subsection{Bilateral Rotations}\label{app:Bilateral_rots}
We review the bilateral rotations described in \cite{BDSW96} and how Bell states transform under such rotations.
Let $R_x = e^{-i\frac{\pi}{4}X},R_y = e^{-i\frac{\pi}{4}Y},R_z = e^{-i\frac{\pi}{4}Z},$ and generally $R_Q=e^{-i\frac{\pi}{4}Q} = \frac{1}{\sqrt{2}}(I-iQ)$ for $Q\in\{X,Y,Z\}$. A bilateral rotation is applied on both Alice and Bob's particles: $B_x = R_x \otimes R_x,B_y = R_y \otimes R_y,B_z = R_z \otimes R_z. $
Using Lemma \ref{lem:ricochet}, we can observe the effect of carrying out such a bilateral rotation on a Bell state affected by a Pauli error $P$ as follows:
\begin{align*}
    (R_Q\otimes R_Q)(\mathbbm{I}\otimes P) \ketbra{B_{00}}(\mathbbm{I}\otimes P)(R_Q^\dagger\otimes R_Q^\dagger) =
    (\mathbbm{I}\otimes e^{-i\frac{\pi}{4}Q} P e^{-i\frac{\pi}{4}Q^T}) \ketbra{B_{00}}(\mathbbm{I}\otimes e^{i\frac{\pi}{4}Q^T} P e^{i\frac{\pi}{4}Q}).
\end{align*}
The bilateral rotation then transforms the state according to Table \ref{tab:bilateral_rotation}.
\begin{table}[htb!]
    \centering
    
\begin{tabular}{|l||*{4}{c|}}\hline
\backslashbox{$P$}{$Q$}
&\makebox[3em]{$\mathbbm{I}$}&\makebox[3em]{$X$}&\makebox[3em]{$Y$}
&\makebox[3em]{$Z$}\\\hline\hline
$\mathbbm{I}$ &$\mathbbm{I}$ &$X$&$\mathbbm{I}$&$Z$\\\hline
$X$ &$X$&$\mathbbm{I}$&$Z$&$X$\\\hline
$Y$ &$Y$&$Y$&$Y$&$Y$\\\hline
$Z$ &$Z$&$Z$&$X$&$\mathbbm{I}$\\\hline
\end{tabular}
\caption{$e^{-i\frac{\pi}{4}Q} P e^{-i\frac{\pi}{4}Q^T}$}
    \label{tab:bilateral_rotation}
\end{table}

Consider the state \begin{equation}
    \label{eqn:BellPauliNoise}
    \rho=p_I \ketbra{B_{00}}+p_X X_B \ketbra{B_{00}}X_B+p_Y Y_B \ketbra{B_{00}}Y_B+p_Z Z_B\ketbra{B_{00}}Z_B 
\end{equation}that results from sending half of a Bell state $\ketbra{B_{00}}$ across a Pauli channel.
Applying bilateral rotations to such a state permutes the Pauli operators. For example, applying $B_x$ to $\rho$, we get $B_x \rho B_x^\dagger=p_I X_B\ketbra{B_{00}}X_B+p_X \ketbra{B_{00}}+p_Y Y_B \ketbra{B_{00}}Y_B+p_Z Z_B \ketbra{B_{00}}Z_B.$ Thus, the channel vector is permuted from $[p_I,p_X,p_Y,p_Z]$ to $[p_X,p_I,p_Y,p_Z].$

\subsection{Symplectic representation} \label{app:symplectic}
A Pauli string of length $n$ is a tensor product of Pauli matrices acting on $n$ qubits. Upto a global phase~($\pm i, \pm 1$), such a Pauli string can be represented by a binary string of length $2n$, where the first $n$ bits represent the $X$-component, such that the $i^{th}$ bit is 1 only if $P$ has an $X$ or $Y$ in the $i^{th}$ position and 0 otherwise. The last $n$ bits represent the $Z$-component, such that the $(i+n)^{th}$ bit is 1 only if $P$ has an $Z$ or $Y$ in the $i^{th}$ position.   In a similar manner an $n$-pair Bell state is denoted by $\ket{B_s}=\ket{B_{x_1,z_1}}\otimes \ket{B_{x_2,z_2}},\dots \otimes\ket{B_{x_n,z_n}}$, where $s=(x_1,z_1),\dots,(x_n,z_n)$ is a binary string of length $2n$.
This mapping between Paulis and binary strings is called the symplectic representation. Let $s$ be the symplectic representation of $P$ and $t$ be the symplectic representation of $Q$. 
%Determining if two Pauli strings $P,Q\in \mathcal{P}^n$ commute is expressed in the symplectic representation by the symplectic inner product between their symplectic representations $s,t$.
%
The symplectic inner product between $s$ and $t$ determines whether $P$ and $Q$ commute or anti-commute.
Let $s_{x}, t_{x}$ be the $X$-part of the symplectic representation of $P$, $Q$ respectively, and $s_z, t_{z}$ be the $Z$-parts of $P$, $Q$ respectively. 
$P,Q$ commute if and only if the symplectic inner product: $(s_{x}\cdot t_{z}+s_{z}\cdot t_{z})\bmod 2 = 0 $, else they anti-commute~\cite{Rengaswamy_2018}.
This establishes a direct connection between linear parity checks applied to binary strings
and stabilizer measurements. 

\newcommand{\sB}{\textbf{s}}
\newcommand{\tB}{\textbf{t}}
\newcommand{\rB}{\textbf{r}}
\newcommand{\fB}{\textbf{f}}
\newcommand{\zrB}{\textbf{0}}

\subsection{PPT Bell-diagonal states}\label{app:bellppt}
Here we recap an expression for testing whether a Bell-diagonal state is PPT, as discussed in \cite{divincenzo2002quantum}. 

A Bell-diagonal matrix $M$ of dimension $2^{2n} \times 2^{2n}$ is diagonal in the Bell basis; $$M = \sum_{s \in \{0,1\}^{2n}} \alpha_s \ketbra{B_s}$$
Note that ${B_{00}}^{T_A} =\frac{1}{2}(B_{00}+B_{01}+B_{10}-B_{11}) $, so $$({B_{00}}^{T_A})^{\otimes n} =\frac{1}{2^n} \sum_{k\in\{0,1\}^{2n}} (-1)^{N_{11}(k)} \ketbra{B_k},$$
where $N_{11}(k)$ denotes the number of occurrences of $\ket{B_{11}}$ in $\ket{B_k}.$
Note that $\ket{B_s} = I_A \otimes P_{s,B} \ket{B_{00}}^{\otimes n},$ for some Pauli string $P_s$.

The partial transpose of $M$ is 
$ M^{T_A} = \frac{1}{2^n}\sum_{s,k\in\{0,1\}^{2n}}\alpha_s (-1)^{N_{11}(k)} \ketbra{B_{s\oplus k}}.$
$M$ is PPT iff 
$\forall m \in \{0,1\}^{2n},\bra{B_m} M^{T_A} \ket{B_m} \geq 0 $, or equivalently, \begin{equation} \forall m \in \{0,1\}^{2n}, \sum_{s\in\{0,1\}^{2n}} \alpha_s (-1)^{N_{11}(s\oplus m)}\geq 0.\label{eq:BellPPTCondition}\end{equation}

\subsection{Useful definitions}
Define a commutation function as follows: $$c(P_1,P_2)= \begin{cases} 1 & \textnormal{if } P_1,P_2 \textnormal{ commute},\\
-1 & \textnormal{otherwise.}
\end{cases}.$$

\section{Extended Literature Review}
\label{app:lit_review}

\subsection{Parity checks flavors: AEM \& BPM}\label{app:AEMBPM}
\paragraph{Appended E-bit Measurements (AEM)\cite{AsympAdaptive06}}
In the case of AEMs, Alice and Bob share an extra noiseless Bell pair. We may assume they share a state of the form $$
\sum_{s\in\{0,1\}^{2n}}p_s \ketbra{B_s}\otimes \ketbra{B_{00}},$$
where the first $2n$ qubits are the shared noisy Bell pairs at the output of the Pauli channel, and the final noiseless pair is the appended e-bit.
Let $f_x(s) = (x\cdot s) \bmod 2$, for some $x \in \{0,1\}^{2n}$, be any linear binary parity check. The value of this check can be determined in two steps.
\begin{enumerate}
    \item 
By carrying out local Cliffords (bilateral XOR (Fig.\ref{fig:bilateral-XOR})), local Paulis and bilateral rotations \ref{app:Bilateral_rots} \cite{BDSW96}), Alice and Bob can map their shared state to 
$$
\sum_{s\in\{0,1\}^{2n}}p_s \ketbra{B_s}\otimes \ketbra{B_{f_x(s),0}}.$$

Alice and Bob each measure their respective halves of the appended e-bit in the $Z$-basis, Alice can then share her measurement outcome with Bob. 

\item Using classical communication, Bob learns Alice's measurement outcome. By comparing Alice's outcome $a$ with his own $b$, Bob can determine the value of $f_x(s)$ as $a\oplus b.$ This can be seen as $f_x(s)=0$ implies that the appended e-bit is in the state $\frac{\ket{00}+\ket{11}}{\sqrt{2}}$ resulting in Alice's and Bob's measurement outcomes being equal. On the other hand, $f_x(s)=1$ implies that the appended e-bit is in the state $\frac{\ket{01}+\ket{10}}{\sqrt{2}}$ resulting in Alice's and Bob's measurement outcomes being unequal.
\end{enumerate}

\paragraph{Bilateral Pauli Measurements (BPM)\cite{AsympAdaptive06}}
In a BPM, we also proceed similarly in two steps; the first is applying local Cliffords and measurements, and the second is classical communication.
\begin{enumerate}
\item In the case of a BPM, Alice and Bob carry out the same projective measurement on their respective halves of the noisy Bell pairs. 
We study the effect of these measurements on each of the states $(\mathbbm{I}_A\otimes P_B) \ket{B_{00}}^{\otimes n}$ in the decomposition of the noisy state (\ref{eqn:noisyPauliState}).
Let $Q$ be a Pauli operator, then projectors  
% $\Pi_{+} = \frac{\mathbbm{I}+Q}{2}$ and 
% $\Pi_{-} = \frac{\mathbbm{I}-Q}{2}$
%
$\Pi_{+} = (\mathbbm{I}+Q)/2$ and 
$\Pi_{-} = (\mathbbm{I}-Q)/2$
form elements of this projective measurement. Let $a\in \{+1,-1\}$ be the outcome of Alice's measurement, $b\in\{+1,-1\}$ be the outcome of Bob's measurement, and $\tilde{c} = c(Q,P) $. 
The post-measurement state is 
$$
\frac{\mathbbm{I}_A+a Q_A}{2} \otimes \frac{\mathbbm{I}_B+b Q_B}{2} (\mathbbm{I}_A\otimes P_B) \ket{B_{00}}^{\otimes n}.$$
Let $\hat{Y}(P)$ be a function for the number of $Y$'s in a Pauli string, for example, $\hat{Y}(I\otimes Y \otimes X \otimes Y\otimes Z) = 2$.
Using the Ricochet property of Lemma \ref{lem:ricochet}, and the fact that $Q^T = (-1)^{\hat{Y}(Q)} Q, $ we may rewrite the post-measurement state as: 
\begin{equation}\label{eqn:postmmt}
\mathbbm{I}_A\otimes\left(\frac{1+ab\tilde{c}(-1)^{\hat{Y}(Q)}}{2}\right)P_B \frac{\mathbbm{I}_B+b\tilde{c}Q_B}{2}\ket{B_{00}}^{\otimes n}.
\end{equation}

\item The probability that Alice and Bob's measurement outcomes satisfy $ab=\tilde{c}(-1)^{\hat{Y}(Q)}$ is 1, so they can always deduce the value of $\tilde{c}(-1)^{\hat{Y}(Q)}$ by classically communicating their measurement outcomes.
\end{enumerate}
By tracing out the $A$ systems of the state in (\ref{eqn:postmmt}), we get a state $\propto P\left(\mathbbm{I}_B+b\tilde{c}Q_B\right)$ with $2^{n-1}$ non-zero eigenvalues. Thus, the Schmidt rank of the state (\ref{eqn:postmmt}) is $2^{n-1}$, implying that the measurement costs 1 of the shared noisy Bell-pairs. 

Note that in contrast with the AEM case, BPMs have a `backaction' that alters the state shared between Alice and Bob. Such a change of the shared state must be accounted for in entanglement distillation protocols. %TODO: degeneracy 

\subsection{Interpolation of Recurrence and Hashing \cite{VV05}}
The protocol of \cite{VV05} makes use of two main ideas; the first is partial breeding and the second is bilateral Pauli measurements. 
Partial breeding can be thought of as an interpolation of recurrence and hashing in the sense that in recurrence a finite-size parity check that consumes a Bell pair out of each block of two Bell pairs is carried out, determining the value of the check for each of the blocks, while hashing exactly determines the entire Pauli string at an entropic cost, rather than a fixed cost per check. In partial breeding, the same information as a finite-size check (such as recurrence) is determined, but since the check is done in an asymptotic manner like hashing, the cost depends on the entropy of the outcome distribution. Thus, the partial breeding step of the protocol of \cite{VV05} determines asymptotically a linear parity check $s\cdot r$, where $r$ is $n/2$ copies of $[1010]$, determining which blocks of two Bell pairs are affected by a single $X$ error. 
On the other hand, the obtained state at the following stages of the protocol may have a negative breeding/hashing rate despite being distillable. To extract the leftover entanglement, \cite{VV05} suggest carrying out a bilateral Pauli measurement, similar to a finite size recurrence step. The back-action and degeneracy effects due to this finite size check manage to squeeze out a state with less entropy and thus having a positive hashing rate.

Here, we explain in more detail the steps of the protocol of \cite{VV05}. \label{app:VV}

\begin{itemize}

\item The protocol starts by applying `partial breeding', where, similarly to breeding, many asymptotic (infinite-size) AEM random parity checks are applied. However, in contrast with breeding, instead of applying  checks to completely identify  $s$, only enough checks are applied to determine the values of $s_{4i} \oplus s_{4i+2} $ for all $i$.
In other words, instead of completely determining the Pauli error $P_B$, partial breeding is used to find out, for every pair of noisy qubits $2k,2k+1$, if the error $P_{2k,B}\otimes P_{2k+1,B} $  commutes or anticommutes with $Z_{2k}\otimes Z_{2k+1}.$ 

Let $\rho_{2}$ be the state of the noisy Bell pairs $2k$ and $2k+1$,   
\begin{equation}
    \rho_{2} = \sum_{P_{2k,B},P_{2k+1,B} \in \mathcal{P}^{\otimes 2}}\textnormal{Pr}[P_{2k,B}\otimes P_{2k+1,B}] P_{2k,B}\otimes P_{2k+1,B} \ketbra{B_{00}}^{\otimes 2} P_{2k,B}\otimes P_{2k+1,B},
    \label{eq:belldiagonal2qubit}
\end{equation}
where $\textnormal{Pr}[P_{2k,B}\otimes P_{2k+1,B}]$ is the probability that a certain Pauli error $P_{2k,B},P_{2k+1,B} \in \mathcal{P}^{\otimes 2}$ occurs 
on Bob's parts of the $2k^{th}, (2k+1)^{th}$ noisy Bell pairs, according to the channel model.
Once the partial breeding checks are complete, Alice and Bob can determine 
if 
$P_{2k,B}\otimes P_{2k+1,B}  \in \mathcal{P}_{\text{even}}=\{IZ,ZI,XX,YY,ZZ,II,XY,YX\}$. Let $p_{\text{even}}$ denote the probability that $P_{2k,B}\otimes P_{2k+1,B}  \in \mathcal{P}_{\text{even}}$.
The number of noiseless Bell pairs to be consumed in the partial breeding steps, which is the number of AEM checks required to determine for all $k$ if $P_{2k,B}\otimes P_{2k+1,B}\in \mathcal{P}_{\text{even}}$ is $h_b(p_{\text{even}})$ clean Bell pairs per 2 noisy pairs. 

\item For every pair of noisy qubits $2k,2k+1$, if indeed $P_{2k,B}\otimes P_{2k+1,B}\in \mathcal{P}_{\text{even}} $, then the obtained state $$\rho_{\textnormal{even}} = \frac{1}{p_{\text{even}}}\sum_{P_1\otimes P_2 \in \mathcal{P}_{\text{even}}} \text{Pr}\left[P_{1,B}\otimes P_{2,B}\right] P_{1,B}\otimes P_{2,B} \ketbra{B_{00}}^{\otimes 2}P_{1,B}\otimes P_{2,B}$$ is then fed into the hashing protocol, which gives the yield $2-S(\rho_{\textnormal{even}})$, per 2 noisy pairs.

\item Otherwise, for pairs of Bell pairs where the 2-qubit error $P_{2k,B}\otimes P_{2k+1,B} $ anticommutes with $Z_{2k}\otimes Z_{2k+1},$   $P_{2k,B}\otimes P_{2k+1,B} \in \mathcal{P}_{\text{odd}}=\{XZ,ZX,XI,YI,IX,IY,YZ,ZY\}$. This occurs with probability $p_{\text{odd}} =1-p_{\text{even}}.$ The state in this case is 
$$\rho_{\textnormal{odd}} = \frac{1}{p_{\text{odd}}}\sum_{P_1\otimes P_2 \in \mathcal{P}_{\text{odd}}} \text{Pr}\left[P_{1,B}\otimes P_{2,B}\right] P_{1,B}\otimes P_{2,B} \ketbra{B_{00}}^{\otimes 2}P_{1,B}\otimes P_{2,B}.$$
In this case, a BPM (finite-size) check is applied to the state $\rho_{\textnormal{odd}}$ to determine if $P_{1,B}$ commutes with $Z_1$. In this case, determining this commutation relation is not done via many random infinite size measurements with appended e-bits asymptotically, but through a finite size check that applies a projective measurement onto the $+1$ and $-1$ eigenstates of $Z$ on both Alice and Bob's particles of the first noisy pair. Such a finite size check collapses one of the noisy Bell pairs of the state $\rho_{\text{odd}}.$ Moreover, since the errors $P_1 \otimes P_2$ and $P_1Z \otimes P_2$ differ by a stabilizer, they have the same effect on the state, and are mapped to the same state, leading to entropy reduction, as pointed out by \cite{AsympAdaptive06}.   
 After this measurement, we end up with a rank-2 state.

Given that the partial breeding step indicated that the 2-qubit error anticommutes with $ZZ,$ then \begin{itemize}
    \item[\ding{228}] $P_{1,B}$ commutes with $Z_1$ with probability $\frac{p_0}{p_{\text{odd}}}$, where $p_0 = p_{ZX}+p_{IX}+p_{ZY}+p_{IY}$. In this case, the state becomes $$\rho_{\text{odd},0} = \frac{p_{ZX}+p_{IX}}{p_{0}} 
 X_B \ketbra{B_{00}}X_B+\frac{p_{ZY}+p_{IY}}{p_0}Y_B \ketbra{B_{00}}Y_B.$$
 This state can then be passed on to the hashing protocol achieving a yield of $1-S(\rho_{\text{odd},0}).$

 \item[\ding{228}] $P_{1,B}$ anticommutes with $Z_1$ with probability $\frac{p_1}{p_{\text{odd}}}$, where $p_1 = p_{XZ}+p_{XI}+p_{YZ}+p_{YI}$. In this case, the state becomes $$\rho_{\text{odd},1} = \frac{p_{XZ}+p_{YZ}}{p_{1}} 
 Z_B \ketbra{B_{00}}Z_B+\frac{p_{XI}+p_{YI}}{p_1}\ketbra{B_{00}}.$$
  This state can then be passed on to the hashing protocol achieving a yield of $1-S(\rho_{\text{odd},1}).$
\end{itemize}
\end{itemize}

In summary, the protocol starts by applying partial breeding to $n/2$ copies of a Bell diagonal state over 2 Bell pairs $\sigma$, which can be decomposed as 
$
    \sigma = p_{\text{even}}\rho_{\text{even}} \oplus 
    p_{\text{odd}} \rho_{\text{odd}},
$
where the direct sum symbol $\oplus$ expresses that $\rho_{\text{even}}$ and 
$\rho_{\text{odd}}$ are orthogonal. Then, the state $$\sigma^{\otimes (n/2)}=
\sum_{f \in \{\text{even},\text{odd}\}^{n/2}}\textnormal{Pr}[f] \bigotimes_{f_i \in f} \rho_{f_i}.$$ After partial breeding we learn an $n/2$-bit string indicating which pairs have an error $\in \mathcal{P}_{\text{even}}$, and which have an error $\in \mathcal{P}_{\text{odd}}$. Once we learn $f$, the pairs in the even case are passed on to the hashing protocol obtaining $1-S(\rho_{\text{even}})/2$ per input Bell pair, while those in the odd case are passed on to another round of a BPM check, then to hashing. 
Since the even case occurs with probability $p_{\text{even}}$, and since the outcome of final BPM check is 0 with probability $p_0$ yielding a state $\rho_{\text{odd},0}$ and is $1$ with probability $p_1 = p_{\text{odd}}-p_0 = 1-p_{\text{even}}-p_0$ yielding a state $\rho_{\text{odd},1}$, the overall yield of the protocol is given by 

\begin{align}\text{Yield}_{VV}(\rho_2)=
p_{\text{even}}\left[1- \frac{S(\rho_{\text{even}})}{2}\right]-\frac{h_b(p_{\text{even}})}{2}+\frac{p_0}{2} \left[1-S\left(\rho_{\text{odd},0}\right)\right]+\frac{p_1}{2}  \left[1-S\left(\rho_{\text{odd},1}\right)\right].
\end{align}

\subsection{Adaptive Entanglement Purification Protocol $\texttt{AEPP}^*(4)$}

In \cite{leung2007adaptive}, the authors propose an adaptive protocol that relies on performing a BPM (finite-size) check on $m$ noisy Bell pairs to determine if the Pauli error affecting the state commutes or anticommutes with $Z^{\otimes m}$. 
In case the error commutes with $Z^{\otimes m}$, the state is passed on to the hashing protocol. Otherwise, further shorter finite-size checks are applied to localize the error. 

A modified version of this protocol is also suggested in the same work \cite{leung2007adaptive}, where instead of switching directly to hashing when the error commutes with $Z^{\otimes m}$, a check in the complementary basis is performed first. The authors use these ideas to devise the $\texttt{AEPP}^*(4)$, which we review in detail. 

Consider a noisy Bell-diagonal state, where 4 noisy Bell-pairs are shared between Alice and Bob, 
$ \sum_{s \in\{0,1\}^8} p_{s} \ketbra{B_{s} }, $ with $s$ being a binary string of the form $s=(x_1,z_1),(x_2,z_2),(x_3,z_3),(x_4,z_4)$.
\begin{itemize}
\item As a first step, a BPM finite-size check is carried out on 4 noisy Bell pairs. This check determines if the error affecting the state commutes with $Z^{\otimes 4}$. The check is carried out by applying the circuit in Figure \ref{fig:ZZZZ}, which consists of several applications of bilateral XORs. 
The effect of the circuit is to map $s$ to $s^\prime=(x_1,z_1\oplus z_4),(x_2,z_2\oplus z_4),(x_3,z_3\oplus z_4),(x_1\oplus x_2\oplus x_3\oplus x_4,z_4).$
\begin{figure}[htb]
\centering
\begin{quantikz}
\makeebit[angle=-60]{} &              &\ctrl{6}&&&&& \\  &\gate{x_1,z_1}  &&\ctrl{6} &&&& \\
\makeebit[angle=-60]{} && && \ctrl{4}&&&\\
&\gate{x_2,z_2}&&&&\ctrl{4}&&\\
\makeebit[angle=-60]{} &&&&&&\ctrl{2} & \\
&\gate{x_3,z_3}&&&&&&\ctrl{2} \\
\makeebit[angle=-60]{} &&\targ{} &&\targ{}&&\targ{}&&\meter{}  \\
&\gate{x_4,z_4}&&\targ{} &&\targ{}&&\targ{}&\meter{} 
\end{quantikz}
\caption{$ZZZZ$ check.}
\label{fig:ZZZZ}
\end{figure}
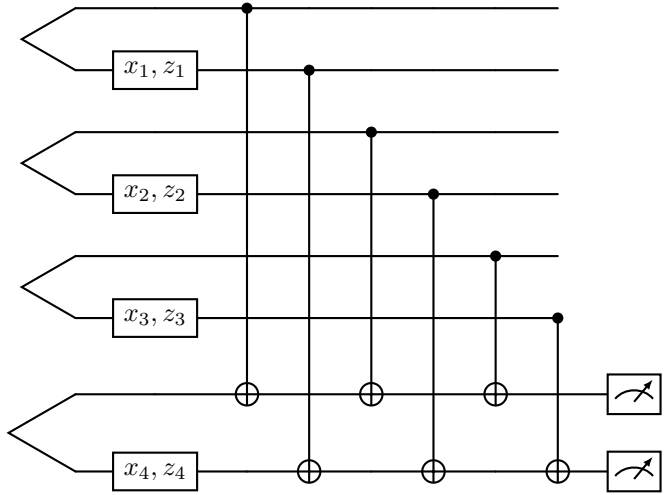

\item If the error commutes with $ZZZZ$, then the syndrome bit $b_1=x_1\oplus x_2\oplus x_3\oplus x_4$ corresponding to this check takes the value $b_1=0.$ If $b_1=0$, then the next step is to check for errors in the complementary basis, so we measure $XXX$ on the state resulting from the previous check by applying the circuit in Figure \ref{fig:XXXX} to the 3 surviving pairs. This circuit maps the Bell string $s^\prime$ to 
$s^{\prime\prime} =(x_1\oplus x_3,z_1\oplus z_4),(x_2\oplus x_3,z_2\oplus z_4),(z_1\oplus z_2\oplus z_3\oplus z_4,x_3),(x_1\oplus x_2\oplus x_3\oplus x_4,z_4). $
The syndrome bit corresponding to this check is $b_2 = z_1\oplus z_2\oplus z_3\oplus z_4$.
This is equivalent to having measured the $XXXX$ check on the initial state (prior to the $ZZZZ$ check). Note that $ZZZZ,XXXX$ are the stabilizer checks of the well-known $\llbracket 4,2,2\rrbracket$ code. 

After measuring the $ZZZZ$, $XXXX$ checks, and obtaining the syndrome bits $b_1,b_2$, the post-measurement state, under the assumption that all 4 noisy Bell pairs are affected by i.i.d. Pauli noise with the probability of a Pauli error $X^x Z^z$, $x,z \in\{0,1\}$, is proportional to 
\begin{equation}
\sum_{r_x,r_z,t_x,t_z \in \{0,1\}^4} q^{b_1,b_2}(r_x,r_z,t_x,t_z)\ketbra{B_{r_x,r_z}}\otimes\ketbra{B_{t_x,t_z}}, \label{eq:422Outputstate}
\end{equation}
where
\begin{equation}
q^{b_1,b_2}(r_x,r_z,t_x,t_z) = \sum_{x_1,z_1}p(x_1,z_1)p(x_1\oplus r_x\oplus t_x,z_1\oplus r_z\oplus t_z)p(x_1\oplus r_x,b_2\oplus z_1\oplus t_z)p(b_1\oplus x_1\oplus t_x,z_1\oplus b),\label{eq:q422}\end{equation}
\item If the error anticommutes with $ZZZZ$, with $b_1=1$, then instead of applying the $XXXX$ check, another $Z$-type check is carried out to localize the error that caused the first check to fire. To do so,we measure the check $ZZII.$
Let $b_2^\prime$ be the syndrome bit obtained upon measuring $ZZII$, the leftover post-measurement state (on the first and third pairs) is proportional to the following product state: 
\begin{align*}\sum_{x_1,\Tilde{z}_1,z_1 \in\{0,1\}}p(x_1,z_1)p(x_2=b_2^\prime\oplus x_1,z_2=\Tilde{z}_1\oplus z_1)\ketbra{B_{x_1,\Tilde{z}_1}}\\\otimes\sum_{x_3,\Tilde{z}_3,z_3 \in\{0,1\}}p(x_3,z_3)p(x_4=1\oplus b_2^\prime\oplus x_3,z_4=\Tilde{z}_3\oplus z_3)\ketbra{B_{x_3,\Tilde{z}_3}}\end{align*}
When $b_2^\prime = 0 $, the first pair is accepted pair, while the third pair becomes undistillable, and is therefore discarded. When $b_2^\prime = 1 $, the first pair becomes undistillable and the third pair is kept. The output state can then be passed on to other protocols. In either case, the accepted pair has a probability distribution representing an accepted pair in one step of a $Z$ recurrence protocol, i.e., \begin{equation}\biggl[p_I^\prime =\frac{p_I^2+p_Z^2}{p_{\text{accept}}}, p_X^\prime = \frac{p_X^2+p_Y^2}{p_{\text{accept}}}, p_Y^\prime =\frac{2p_Yp_X}{p_{\text{accept}}}, p_Z^\prime =\frac{2p_Ip_Z}{p_{\text{accept}}}\biggr],\label{eq:Daccept}\end{equation} where, $p_{\text{accept}} = (p_I+p_Z)^2+(p_X+p_Y)^2$.

%To Do: show plot of intermediate range of parameters where AEPP is better than recurrence 

\end{itemize}

\begin{figure}[htb]
\centering
\begin{quantikz}
\makeebit[angle=-60]{} &&\targ{}&&&& \\  &\gate{x_1,z_1}  &&\targ{} &&& \\
\makeebit[angle=-60]{} && && \targ{}&&\\
&\gate{x_2,z_2}&&&&\targ{}&\\
\makeebit[angle=-60]{} &&\ctrl{-4} &&\ctrl{-2}&&\gate{H}&\meter{}  \\
&\gate{x_3,z_3}&&\ctrl{-4} &&\ctrl{-2}&\gate{H}&\meter{} 
\end{quantikz}
\caption{$XXXX$ check following $ZZZZ$ check.}
\label{fig:XXXX}
\end{figure}
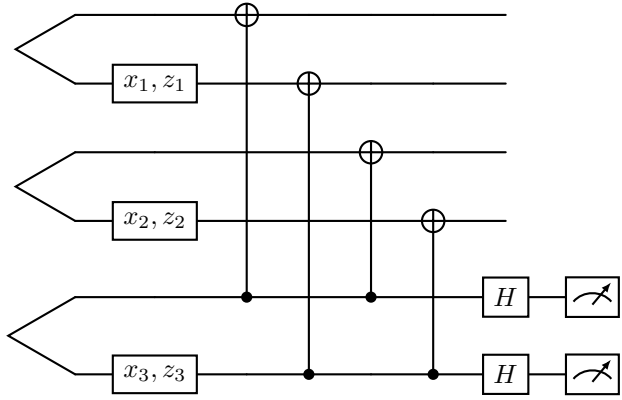

\section{Useful lemmas}
\label{app:proofs_lemmas}

The following lemmas will be useful for analyzing different distillation protocols.
\begin{lemma}(Entropy reduction by grouping)~\cite{AsympAdaptive06}\label{lem:entropy_reduction}
    $$\forall x,y\geq 0, -x\log_2 x-y \log_2 y +(x+y)\log_2(x+y)=(x+y)h_b\left(\frac{x}{x+y}\right)\geq 0.$$
\end{lemma}

\begin{lemma}(Inequality of ratios of arithmetic mean to geometric mean of pairs of disparate probabilities vs pairs of similar probabilities)
\label{lem:ramgm}
Let $a,b,c,d$ be the elements of a probability distribution satisfying $a=\frac{1}{2}+\epsilon$, $\epsilon > 0$, and $a\geq c \geq d \geq b.$
Then, 
$$cd(a^2+b^2)-ab(c^2+d^2)\geq 0.$$
\end{lemma}
\begin{proof}
We may factorize the l.h.s. of the inequality as $(ad-bc)(ac-bd)$. Due to the ordering of the variables $a\geq c \geq d \geq b,$ we have $ad\geq bc$ and $ac\geq bd$. Thus, both factors of the l.h.s. satisfy $(ad-bc)\geq 0$ and $(ac-bd)\geq 0$ proving the inequality.

\end{proof}

\begin{lemma}(Entropy of a direct sum)\label{lem:entropy_of_direct_sum}
Let $ \rho = \sum_{i \in \mathcal{S}_0} \lambda_i \ketbra{\psi_i}+\sum_{i \in \mathcal{S}_1} \lambda_i \ketbra{\psi_i},$
where $\mathcal{S}_0,\mathcal{S}_1$ are two disjoint sets of indices, $\sum_{i \in \mathcal{S}_0\cup\mathcal{S}_1} \lambda_i=1$ and $\bra{\psi_k}{\psi_{k^\prime}}\rangle=\delta_{kk^\prime} \forall k,k^\prime \in \mathcal{S}_0\cup\mathcal{S}_1.$
Moreover, let 
$$ p_x = \sum_{i \in \mathcal{S}_x} \lambda_i, 
\quad 
\Tilde{\rho}_x =\sum_{i \in \mathcal{S}_x} \frac{\lambda_i}{p_x}\ketbra{\psi_i}, x\in\{0,1\}, $$
i.e., $\rho = p_0 \Tilde{\rho}_0\oplus p_1\Tilde{\rho}_1.$
Then, $$S(\rho) = p_0 S(\Tilde{\rho}_0)+p_1 S(\Tilde{\rho}_1)+h_b(p_0).$$
\end{lemma}
\begin{proof}
Extend $\{\ket{\psi_i}\}$ to a full orthonormal basis, and consider the isometry 
$$V =\sum_{i}\ket{f(i)}\ket{\psi_i}\bra{\psi_i},$$ with $f(i) = 0$ if $i\in \mathcal{S}_0$, $f(i) = 1$ if $i\in \mathcal{S}_1$ and $f(i) = 2$ if $i$ not in $\mathcal{S}_0$ or $\mathcal{S}_1$. Note that $V$ satisfies $V^\dagger V ={I}.$ 
Since entropy does not change under the application of an isometry $S(\rho)=S(V\rho V^\dagger),$ where $V\rho V^\dagger$ is a classical-quantum state given by $$ V\rho V^\dagger = \sum_{x \in \{0,1\}}p_x\ketbra{x}\otimes \tilde{\rho}_{x}$$

Due to this identity for joint entropy of a classical quantum state $S(B,X) = S(B|X)+S(X)$, we may write 
\begin{align*} S(\rho)=S(V\rho V^\dagger)=\sum_x p_x S(\Tilde{\rho}_x)+h_b(p_0).
\end{align*}
\end{proof}

\section{Proofs of theorems providing performance guarantees for Greedy recurrence}
\label{app:proofs_thms_greedy} 
\guaranteedimprovement*
\begin{proof}
    Since we focus on the range of input states with $p_I> 1/2$, we may write $p_I = 1/2+\epsilon$ for $0< \epsilon < 1/2. $

    By construction of the greedy recurrence scheme, at each step we apply $Q$-recurrence, with $Q$ being the Pauli channel component with the smallest probability $p_Q$ i.e., $p_Q \leq p_{Q^\prime}, \forall Q^\prime \in \{I,X,Y,Z\} $. Then, $p_Q=(1/2-\delta)/3$ with $\epsilon \leq \delta \leq 1/2$.  
    Using eq.~\eqref{eqn:accepted_distribution}, we find that the improvement is $$ p_I^\prime -p_I = \epsilon \frac{\alpha-\gamma}{\beta+\gamma},$$ with
    \begin{align*}
    \alpha =& 81+648\delta\\
    \beta =& 405-162\delta\\
      \gamma =& 162 \delta^{2} - 972 \delta \epsilon + 1458 \epsilon^{2} + 486 \epsilon, 
    \end{align*}
    For $0<\epsilon<1/2,\epsilon\leq \delta\leq 1/2$, $\beta>0,\gamma>0$. 
    Assuming a strict inequality: $\epsilon<\delta$, 
    we can show 
    \begin{align*}
     \gamma{<}&486 \epsilon +486 \delta \epsilon +162 \delta^2\\
        <&648 \delta +162\delta(4\delta-1)\\
        {{\leq}}&648\delta+81 = \alpha.
    \end{align*}
    The first two strict inequalities are due to the assumption that $\epsilon<\delta$, while the last inequality is due to the constraint that $\delta\leq1/2$.
    The previous set of inequalities imply $\gamma < \alpha$. 
    In the case when $\epsilon=\delta$, then the strict inequality $\delta<1/2$ holds in the third step instead of the first and second, due to $p_I<1$. This ensures that in both cases, $\gamma < \alpha$ and $p_I^\prime>p_I.$
    
\end{proof}

\optimalityofgreedyrecurrence*
\begin{proof}
Using eq.~\eqref{eqn:accepted_distribution}, we have
\begin{align*}
    p_{I,Q^\ast}^\prime -p_{I,\Tilde{Q}}^\prime = \frac{p_I^2+(\alpha-p_I)^2}{\alpha^2+(1-\alpha)^2}-\frac{p_I^2+(\beta-p_I)^2}{\beta^2+(1-\beta)^2},
\end{align*}
where $\alpha = p_I+p_Q^\ast$ and 
$\beta = p_I+p_{\Tilde{Q}}$.
To show that $p_{I,Q^\ast}^\prime \geq p_{I,\Tilde{Q}}^\prime$, it suffices to show that the numerator of the above expression is non-negative, as the denominator is a product of sums of squares. 

Thus, we write the numerator as $$f=[p_I^2+(p_I-\alpha)^2][\beta^2+(1-\beta)^2]-[p_I^2+(p_I-\beta)^2][\alpha^2+(1-\alpha)^2].$$
Note that $1/2\leq \alpha \leq \beta \leq p_I\leq 1.$
Thus, $(p_I-\alpha)^2\geq (p_I-\beta)^2$, implying
\begin{equation}\label{ineeq:f1}[p_I^2+(p_I-\alpha)^2]\geq[p_I^2+(p_I-\beta)^2]\geq0.\end{equation}
Moreover, the function $x^2+(1-x)^2$ is increasing over $x\in[1/2,\infty)$. 
This, combined with $1/2\leq \alpha \leq \beta$ implies 
\begin{equation}\label{ineeq:f2}[\beta^2+(1-\beta)^2]\geq [\alpha^2+(1-\alpha)^2]\geq0.\end{equation}
Combining (\ref{ineeq:f1}) and (\ref{ineeq:f2}), we obtain $f\geq 0$.
\end{proof}

\MacVsGreedy*

\begin{proof}
    We prove this by induction. 
    To fix notation, we denote by $t \in \{0,1,\dots\}$ the $t^{th}$ step of recurrence, with input probabilities $\{p_{Q}^{(t)}\},$ and output probabilities $\{p_{Q}^{(t+1)}\},$ computed using eq.~\eqref{eqn:accepted_distribution}, for $Q \in\{I,X,Y,Z\}$.
    The proof uses induction to show that if at the initial step $p_Z$ is equal to the minimum probability, then at following odd numbered steps $\min_{Q\in\{I,X,Y,Z\}} p_Q^{(2t+1)} = p_Y^{(2t+1)}$ and at following even-numbered steps $\min_{Q\in\{I,X,Y,Z\}} p_Q^{(2t)} = p_Z^{(2t+1)}$ implying that $Z,Y,Z,Y,\cdots$ is a valid sequence for Greedy recurrence.\\
    \emph{Base case:}
    At $t=0$, since $p_Z^{(0)}=p_X^{(0)}=p_Y^{(0)}=p/{3}$, i.e., $p_Z^{(0)}\leq p_Y^{(0)}\leq p_X^{(0)}\leq p_I^{(0)}$, we perform a step of $Z$-recurrence.
    At the output we have $$p_I^{(1)} = \frac{(1-p)^2+p^2/9}{p_{\text{pass}_Z}}, \quad p_Z^{(1)} = \frac{2(1-p)p}{3p_{\text{pass}_Z}}, \quad p_X^{(1)} = \frac{2p^2}{9p_{\text{pass}_Z}},\quad p_Y^{(1)} = \frac{2p^2}{9p_{\text{pass}_Z}}.$$
    
    At $t=1$, since $p<1/2$, $p_Y^{(1)}\leq \min \{p_Z^{(1)},p_X^{(1)}\} $, we apply $Y$-recurrence.
    At the output we get $p_X^{(2)} \geq p_Z^{(2)} $ and \begin{align*}
    p_Z^{(2)} &= \frac{2p_X^{(1)}p_Z^{(1)}}{p_{\text{pass}_Y}}= \frac{4[2p^2/9](1-p)p/3}{p_{\text{pass}_Y}p_{\text{pass}_Z}^2} \\ &=\frac{p_Y^{(1)}4(1-p)p/3}{p_{\text{pass}_Y}p_{\text{pass}_Z}}
    \leq\frac{2p_Y^{(1)}p_I^{(1)}}{p_{\text{pass}_Y}} 
    = p_Y^{(2)}.
    \end{align*}
    \\
    \emph{Inductive step:}\\
     At (even-numbered) step $2t$:
     assume $$ p_Z^{(2t)}\leq \min\{ p_Y^{(2t)}, p_X^{(2t)}\},p_Y^{(2t-1)}\leq \min\{p_Z^{(2t-1)}, p_X^{(2t-1)}\}.$$

    These inequalities imply that $Y$-recurrence was performed at step $2t-1$ and $Z$-recurrence will be performed at step $2t$.
    Then, $$p_{\text{pass}_Z}p_X^{(2t+1)}=\left(p_X^{(2t)}\right)^2+\left(p_Y^{(2t)}\right)^2\geq 2 p_X^{(2t)}p_Y^{(2t)} =  p_{\text{pass}_Z} p_Y^{(2t+1)}. $$  By Theorem \ref{thm:guarantee_improvement}
    $$ p_I^{(2t+1)}>p_I^{(2t)}> p_I^{(2t-1)}>1/2.$$
    Moreover,  
    \begin{align*}
        p_{\text{pass}_Z}p_Z^{(2t+1)}=  2 p_I^{(2t)}p_Z^{(2t)} = 4\frac{1}{p_{\text{pass}_Y}^2}\biggl[\left(p_I^{(2t-1)}\right)^2+\left(p_Y^{(2t-1)}\right)^2\biggr]p_X^{(2t-1)}p_Z^{(2t-1)}
    \end{align*}
    
    Using Lemma \ref{lem:ramgm}, with  $a =p_I^{(2t-1)},b=p_Y^{(2t-1)},c=\max\{p_X^{(2t-1)},p_Z^{(2t-1)}\}, d =\min\{p_X^{(2t-1)},p_Z^{(2t-1)}\}$, we have 
    \begin{align*}
        p_{\text{pass}_Z}p_Z^{(2t+1)}&\geq 4\frac{1}{p_{\text{pass}_Y}^2}p_I^{(2t-1)}p_Y^{(2t-1)}\biggl[\left(p_X^{(2t-1)}\right)^2+\left(p_Z^{(2t-1)}\right)^2\biggr]\\&=2p_Y^{(2t)}p_X^{(2t)}= p_{\text{pass}_Y}p_Y^{(2t+1)}
    \end{align*}
    Thus, we have shown that $p_Y^{(2t+1)}\leq \min\{p_Z^{(2t+1)},p_X^{(2t+1)}\}$, and the next step will be a $Y$-recurrence.\\
    
     At (odd-numbered) step $2t+1$:
     assume $$p_Y^{(2t+1)}\leq \min\{p_Z^{(2t+1)}, p_X^{(2t+1)}\}, p_Z^{(2t)}\leq \min\{ p_Y^{(2t)}, p_X^{(2t)}\}.$$

    These inequalities imply that $Y$-recurrence is performed at step $2t+1$ and $Z$-recurrence was performed at step $2t$.
    Then, $$p_{\text{pass}_Y}p_X^{(2t+2)}=\left(p_X^{(2t+1)}\right)^2+\left(p_Z^{(2t+1)}\right)^2\geq 2 p_X^{(2t+1)}p_Z^{(2t+1)} =  p_{\text{pass}_Y} p_Z^{(2t+2)}. $$  By Theorem \ref{thm:guarantee_improvement}
    $$p_I^{(2t+2)}> p_I^{(2t+1)}> p_I^{(2t)}>1/2.$$
    Moreover,  
    \begin{align*}
        p_{\text{pass}_Y}p_Y^{(2t+2)}=  2 p_I^{(2t+1)}p_Y^{(2t+1)} = 4\frac{1}{p_{\text{pass}_Z}^2}\biggl[\left(p_I^{(2t)}\right)^2+\left(p_Z^{(2t)}\right)^2\biggr]p_X^{(2t)}p_Y^{(2t)}
    \end{align*}
    
    Using Lemma \ref{lem:ramgm}, with  $a =p_I^{(2t)},b=p_Z^{(2t)},c=\max\{p_X^{(2t)},p_Y^{(2t)}\}, d =\min\{p_X^{(2t)},p_Y^{(2t)}\}$, we have 
    \begin{align*}
        p_{\text{pass}_Y}p_Y^{(2t+2)}=  2 p_I^{(2t+1)}p_Y^{(2t+1)} &\geq 4\frac{1}{p_{\text{pass}_Z}^2}p_I^{(2t)}p_Z^{(2t)}\biggl[\left(p_X^{(2t)}\right)^2+\left(p_Y^{(2t)}\right)^2\biggr]\\&=2p_Z^{(2t+1)}p_X^{(2t+1)}= p_{\text{pass}_Y}p_Z^{(2t+2)}
    \end{align*}
    Thus, we have shown that $p_Z^{(2t+2)}\leq \min\{p_Y^{(2t+2)},p_X^{(2t+2)}\}$, and the next step will be a $Z$-recurrence.

\end{proof}
\end{document}